\documentclass[onecolumn,showpacs]{revtex4-2}
\usepackage{graphics}
\usepackage{graphicx}
\usepackage{tikz-cd}
\usepackage{dcolumn}
\usepackage{amsmath}
\usepackage{amsthm}
\usepackage{amssymb}
\usepackage{mathtools}
\usepackage{hyperref}

\newtheorem{theorem}{Theorem}[section]
\newtheorem{corollary}[theorem]{Corollary}
\newtheorem{proposition}[theorem]{Proposition}

\newtheorem{definition}[theorem]{Definition}
\allowdisplaybreaks

\newenvironment{remark}[1][Remark]{\begin{trivlist}
\item[\hskip \labelsep {\bfseries #1}]}{\end{trivlist}}

\begin{document}

\title{On some locally symmetric embedded spaces with non-negative scalar curvature and their characterization}

\author{Abbas M \surname{Sherif}$^{1,2,\ast}$, Peter K S Dunsby$^{2,3,\dagger}$, Rituparno Goswami$^{4,\ddagger}$}
\affiliation{$^1$Center for Geometry and Phyics, Institute for Basic Sciences, Pohang University of Science and Technology, 77 Cheongham-ro, Nam-gu, Pohang, Gyeongbuk 37673, South Korea}
\email{abbasmsherif25@ibs.re.kr}

\affiliation{$^2$Cosmology and Gravity Group, Department of Mathematics and Applied Mathematics, University of Cape Town, Rondebosch 7701, South Africa}

\affiliation{$^3$South African Astronomical Observatory, Observatory 7925, Cape Town, South Africa}
\email{peter.dunsby@uct.ac.za}

\affiliation{$^4$Astrophysics Research Center (ARC), School of Mathematics, Statistics and Computer Science, University of KwaZulu-Natal, Private Bag X54001, Durban 4000, South Africa}
\email{vitasta9@gmail.com}

\begin{abstract}
In this work we perform a general study of properties of a class of locally symmetric embedded hypersurfaces in spacetimes admitting a $1+1+2$ spacetime decomposition. The hypersurfaces are given by specifying the form of the Ricci tensor with respect to the induced metric. These are slices of constant time in the spacetime. Firstly, the form of the Ricci tensor for general hypersurfaces is obtained and the conditions under which the general case reduces to those of constant time slices are specified. We provide a characterization of these hypersurfaces, with key physical quantities in the spacetime playing a role in specifying the local geometry of these hypersurfaces. Furthermore, we investigate the case where these hypersurfaces admit a Ricci soliton structure. The particular cases where the vector fields associated to the solitons are Killing or conformal Killing vector fields are analyzed. Finally, in the context of spacetimes with local rotational symmetry it is shown that, only spacetimes in this class with vanishing rotation and spatial twist can admit the hypersurface types considered, and that the hypersurfaces are necessarily flat.  And if such hypersurface do admit a Ricci soliton structure, the soliton is steady, with the components of the soliton field being constants.
\end{abstract} 

\maketitle

\section{Introduction}
\label{intro}

The \textit{Ricci flow} evolution equation, which is given as the set of partial differential equations

\begin{eqnarray}\label{in1}
\frac{\partial\tilde{g}_{\mu\nu}}{\partial t}=-2R_{\mu\nu},
\end{eqnarray}
was introduced by Hamilton \cite{ham1} and used to prove the sphere theorem, with \(\tilde{g}_{\mu\nu}\) denoting a ``time" dependent metric of a manifold with Ricci curvature \(R_{\mu\nu}\). Self similar solutions to \eqref{in1} are the Einstein metrics, while fixed points are Ricci flat manifolds. The flow \eqref{in1} has proved useful in many applications \cite{df1,df2}, including famously aiding the proof of the Poincare conjecture by Perelman  \cite{gp1,gp2,gp3}. A limitation of \eqref{in1} was that it was not geometrical in the sense that it was dependent on the choice of the coordinate parametrising the flow. This led to the modification of \eqref{in1}, known as the \textit{Hamilton-DeTurck} flow, given by
\begin{eqnarray}\label{in2}
\frac{\partial\tilde{g}_{\mu\nu}}{\partial t}=-2R_{\mu\nu}+\mathcal{L}_{\tilde{X}}\tilde{g}_{\mu\nu},
\end{eqnarray}
where \(\tilde{X}\) is a smooth vector field which generates the change of coordinates along the flow (necessarily a diffeomorphism), and \(\mathcal{L}_{\tilde{X}}\) is the Lie derivative along the vector field \(\tilde{X}\). The extra term in \eqref{in2} involving the Lie derivative leaves \eqref{in2} invariant under change of the parameter along the flow. Self similar solutions to \eqref{in2} then solve

\begin{eqnarray}\label{in3}
R_{\mu\nu}+\left(\frac{1}{2}\mathcal{L}_{\tilde{X}}-\varrho\right)\tilde{g}_{\mu\nu}=0,
\end{eqnarray}
where \(\varrho\in\mathbb{R}\). The equation \eqref{in3} is the \textit{Ricci soliton} equation, and naturally generalizes Einstein manifolds. These objects therefore have use in the study of the Ricci geometric flow.

There is a wealth of literature on Ricci flows with most work dedicated to the so-called gradient Ricci solitons. The geometry of these objects have been extensively studied, even more so in the case of three dimensional Riemannian manifolds, where various classification schemes have been provided. In contrast to the \(3\)-dimensional Riemannian case, the Lorentzian case and the case of embedded hypersurfaces of Lorentzian manifolds is less studied, though some interesting results in this regard have been obtained (see \cite{tom1,eric1} and references therein). In dimension two, the Hamilton's cigar soliton \cite{ham1} is the only steady gradient Ricci soliton with positive curvature which is complete (also see \cite{bern1} for more discussion on the classification of \(2\)-dimensional complete Ricci solitons). Complete classification has been provided in dimension three for which the Ricci soliton is shrinking. Due to results by Ivey \cite{iv1}, Perelman \cite{gp1}, Cao \textit{et al.} \cite{cao1} and others, it is known that shrinking Ricci solitons in \(3\)-dimensions are quotients of the \(3\)-sphere \(\mathbb{S}^3\), the cylinder \(\mathbb{R}\times\mathbb{S}^2\) or the Gaussian gradient Ricci soliton on \(\mathbb{R}^3\).

Bryant \cite{bry1} constructed a steady rotationally symmetric gradient Ricci soliton , and Brendle \cite{bre1} showed that this soliton is the only non-flat \(k\)-noncollapsed steady Ricci soliton in dimension three. More recently, classification of the expanding case have been considered under certain integral assumptions on the scalar curvature \cite{cat1,der1} (also, see \cite{oc1,ma1,pet1,wal1,pet2} and associated references for additional results on Ricci solitons).

The subject of geometry of hypersurfaces is an extensively studied area in (pseudo) Riemannian geometry. Besides purely mathematical interests, hypersurfaces play very fundamental roles in various areas of theoretical physics, and a lot of applications can be found especially in the area of Einstein's theory of General Relativity. Cauchy surfaces, for example, are used to formulate the Einstein's equations as an initial value problem. Another area of prominence is black hole horizons (or more generally \textit{marginally trapped tubes}), where these hypersurfaces ``separate" the black hole region from external observers. The geometry and topology of these hypersurfaces have been extensively studied under varying assumptions on the spacetime, albeit from different perspectives.  In fact, there are nice classification results by Hall and Capocci \cite{hac} and Sousa \textit{et al.} \cite{sea}, for 3-dimensional spacetimes which apply to embedded hypersurfaces.

In principle then, the study of Ricci soliton structures on black hole horizons could provide a classification, geometrically, of horizons, which motivates this work. Indeed there might be other geometrical properties one may infer from these hypersurfaces with Ricci soliton structure, including the restrictions placed on the spacetimes in which these hypersurfaces are embedded.

\subsection{Objective of paper}

Effectively employed to study problems in cosmology, astrophysics and perturbation theory (see \cite{pg1,cc1,crb1,gbc1} and references therein), the \(1+1+2\) covariant formalism will be used here to study Ricci soliton structures on a specified class of hypersurfaces embedded in spacetimes admitting a \(1+1+2\) decomposition, as a first step to potential future applications to the study of the geometry of black hole horizons. We will begin by studying some general properties of a class of hypersurfaces, by using the \(1+1+2\) covariant approach, on which we wish to investigate Ricci soliton structure from the equation point of view (without focus on general properties and the dynamics of the Ricci flow geometric equation). In most cases studied, the spacetime is known. Then a choice of the hypersurface is made and studied as a proper subspace of the spacetime. Here, we shall proceed by first prescribing a form of the Ricci tensor for the hypersurfaces, and working out some of the possible restrictions on these hypersurfaces. (In this case the results will be applicable to classes of spacetimes admitting hypersurfaces with the prescribed form of the Ricci curvature.) The most general form of the Ricci tensor for embedded hypersurfaces in \(1+1+2\) decomposed spacetimes is worked out and the conditions reducing the general case to that which is considered throughout this work is specified. We then investigate Ricci soliton structure on these surfaces and see how the nature of the soliton constrains the geometry of the hypersurfaces, as well as physical quantities specifying the hypersurfaces.

\subsection{Outline of paper}

This work has the following outline: in Section \ref{soc1} we present a brief introduction to the formalism of the \(1+1+2\) spacetime decomposition. Section \ref{soc4} presents the form of the Ricci tensor for the class of hypersurfaces to be investigated throughout this work. The associated curvature quantities are then written in terms of the \(1+1+2\) covariant quantities, and the Gauss-Codazzi equations explicitly specified. In Section \ref{soc5} we give a characterisation of the hypersurfaces and study the various constraints induced by properties of the curvature quantities. In Section \ref{soc6} we present a detailed investigation of the case when the considered hypersurfaces admit a Ricci soliton structure. Section \ref{soc7} considers the results from the previous sections in context of a well known class of spacetimes, the locally rotationally symmetric class II spacetimes. Finally, we conclude in Section \ref{soc8} with a summary and discussion of the results obtained in this work.

\section{\(1+1+2\) spacetime decomposition}\label{soc1}

In this section we introduce the \(1+1+2\) covariant splitting of spacetime. We will provide enough details so that those not very familiar with the formalism find it easy to follow the rest of the paper. A great deal of excellent literature exist which details this approach and its applications to relativistic astrophysics and cosmology. The interested reader is referred to \cite{cc1} and references therein.

The procedure for implementing the \(1+1+2\) spacetime decomposition starts with the splitting of the spacetime in the following manner: choose a unit tangent vector field, usually denoted \(u^{\mu}\) satisfying \(u_{\mu}u^{\mu}=-1\), along the observer's congruence. This choice of vector field induces a split of the \(4\)-dimensional metric \(g_{\mu\nu}\) as
\begin{eqnarray}\label{micel1}
h_{\mu\nu}=g_{\mu\nu}+u_{\mu}u_{\nu},
\end{eqnarray}
where the tensor \(h_{\mu\nu}\) projects vectors and tensors, orthogonal to \(u^{\mu}\), onto the \(3\)-space resulting from the \(1+3\) splitting. This projector is the first fundamental form for the \(3\)-space. This splitting introduces two derivatives from the full covariant derivative of the spacetime \(\nabla_{\mu}\):

\begin{enumerate}

\item \textbf{The derivative along the \(u^{\mu}\) direction}: for any tensor \(S^{\mu\dots \nu}_{\ \ \ \ \mu'\dots \nu'}\) one has the derivative

\begin{eqnarray}\label{micel2}
\dot{S}^{\mu\dots \nu}_{\ \ \ \ \mu'\dots \nu'}=u^{\sigma}\nabla_{\sigma}S^{\mu\dots \nu}_{\ \ \ \ \mu'\dots \nu'};
\end{eqnarray}

\item \textbf{The fully orthogonally projected derivative via \(\bf{h^{\mu\nu}}\) on all indices}: for any tensor \(S^{\mu\dots \nu}_{\ \ \ \ \mu'\dots \nu'}\) one has the derivative

\begin{eqnarray}\label{micel3}
D_{\sigma}S^{\mu\dots \nu}_{\ \ \ \ \mu'\dots \nu'}=h^{\mu}_{\ \gamma}h^{\gamma'}_{\ \mu'}\dots h^{\nu}_{\ \delta}h^{\delta'}_{\ \nu'}h^{\varrho}_{\ \sigma}\nabla_{\varrho}S^{\gamma\dots \delta}_{\ \ \ \ \gamma'\dots \delta'}.
\end{eqnarray}

\end{enumerate}
The first derivative is usually associated with the observer's time, called the covariant `\textit{time} derivative. For simplicity it is usually called the dot derivative. The second derivative is called the `\(D\)' derivative.

The next step in the splitting is to make a choice of a normal direction, denoted \(e^{\mu}\), which splits the \(3\)-space, and is orthogonal to \(u^{\mu}\) and satisfies \(e_{\mu}e^{\mu}=1\). This leads to the further splitting of the spacetime metric \(g_{\mu\nu}\) as

\begin{eqnarray}\label{micel4a}
N_{\mu\nu}=g_{\mu\nu}+u_{\mu}u_{\nu}-e_{\mu}e_{\nu},
\end{eqnarray}
where the tensor \(N_{\mu\nu}\) projects to \(2\)-surfaces (referred to as the sheet), vectors and tensors orthogonal to \(e^{\mu}\). The further two derivatives are also introduced:

\begin{enumerate}

\item \textbf{The derivative along the \(e^{\mu}\) direction}: for any \(3\)-tensor \(S^{\mu\dots \nu}_{\ \ \ \ \mu'\dots \nu'}\) one has the derivative

\begin{eqnarray}\label{micel100w}
\hat{S}^{\mu\dots \nu}_{\ \ \ \ \mu'\dots \nu'}=e^{\sigma}\nabla_{\sigma}S^{\mu\dots \nu}_{\ \ \ \ \mu'\dots \nu'};
\end{eqnarray}

\item \textbf{The fully projected spatial derivative on the \(2\)-sheet via \(\bf{N^{\mu\nu}}\) on all indices}: for any \(3\)-tensor \(S^{\mu\dots \nu}_{\ \ \ \ \mu'\dots \nu'}\) one has the derivative

\begin{eqnarray}\label{micel101w}
\delta_{\sigma}S^{\mu\dots \nu}_{\ \ \ \ \mu'\dots \nu'}=N^{\mu}_{\ \gamma}N^{\gamma'}_{\ \mu'}\dots N^{\nu}_{\ \delta}N^{\delta'}_{\ \nu'}N^{\varrho}_{\ \sigma}D_{\varrho}S^{\gamma\dots \delta}_{\ \ \ \ \gamma'\dots \delta'}
\end{eqnarray}

\end{enumerate}
The first case is usually called the `\textit{hat}' derivative, and the second the `\textit{delta}' derivative. The volume element of the \(2\)-surfaces resulting from the further splitting of the \(3\)-space, is the Levi-Civita tensor

\begin{eqnarray*}
\varepsilon_{\mu\nu}\equiv\varepsilon_{\mu\nu\delta}e^{\delta}=u^{\gamma}\eta_{\gamma\mu\nu\delta}e^{\delta},
\end{eqnarray*}
so that contracting \(\varepsilon_{\mu\nu}\) with \(u^{\nu}\) or \(e^{\nu}\) gives zero. The tensor \(\varepsilon_{\mu\nu}\) also satisfies the additional relations:

\begin{eqnarray*}
\begin{split}
\varepsilon_{\mu\nu\delta}&=e_{\mu}\varepsilon_{\nu\delta}-e_{\nu}\varepsilon_{\mu\delta}+e_{\delta}\varepsilon_{\mu\nu},\\
\varepsilon_{\mu\nu}\varepsilon^{\delta\gamma}&=2N_{[\mu}^{\ \delta}N_{\nu]}^{\ \gamma},\\
\varepsilon_{\mu}^{\ \sigma}\varepsilon_{\sigma\nu}&=N_{\mu\nu},\\
\varepsilon_{\mu\nu}\varepsilon^{\mu\nu}&=2,
\end{split}
\end{eqnarray*}

where the square brackets denote the usual antisymmetrization.

Now, let \(S^{\mu}\) be a \(3\)-vector. Then \(S^{\mu}\) may be irreducibly split as

\begin{eqnarray}\label{micel4}
S^{\mu}=\bf{S}e^{\mu}+\bf{S}^{\mu},
\end{eqnarray}
where \(\bf{S}\) is the scalar associated to \(S^{\mu}\) that lies along \(e^{\mu}\), and \(\bf{S}^{\mu}\) lies in the sheet orthogonal to \(e^{\mu}\). Notice that \eqref{micel4} implies the following:

\begin{eqnarray}\label{micel5}
\bf{S}\equiv S_{\mu}e^{\mu}\quad\mbox{and}\quad\bf{S}^{\mu}\equiv S_{\nu}N^{\mu\nu}\equiv S^{\bar{\mu}},
\end{eqnarray}
where the overbar indicates the index projected by \(N^{\mu\nu}\). In a similar manner, projected, symmetric and trace free tensors \(S_{\mu\nu}\) can be irreducibly split as

\begin{eqnarray}\label{micel6}
S_{\mu\nu}=S_{\langle\mu\nu\rangle}=\bf{S}\left(e_{\mu}e_{\nu}-\frac{1}{2}N_{\mu\nu}\right)+2S_{(\mu}e_{\nu)}+\bf{S}_{\mu\nu},
\end{eqnarray}
where we have

\begin{eqnarray*}
\begin{split}
\bf{S}&\equiv e^{\mu}e^{\nu}S_{\mu\nu}=-S_{\mu\nu}N^{\mu\nu},\\
\bf{S}_{\mu}&\equiv e^{\delta}S_{\nu\delta}N_{\mu}^{\ \nu}=\bf{S}_{\bar{\mu}},\\
\bf{S}_{\mu\nu}&\equiv S_{\langle\mu\nu\rangle}\equiv\left(N_{(\mu}^{\ \delta}N_{\nu)}^{\ \gamma}-N_{\mu\nu}N^{\delta\gamma}\right)S_{\delta\gamma},
\end{split}
\end{eqnarray*}
and the round brackets denoting symmetrization. The angle bracket is being used here to denote the projected, symmetric and trace-free parts of a tensor.

Now we can write down the definition of various scalars, vectors and tensors that will appear in the rest of the paper. We have

\begin{eqnarray*}
\begin{split}
\hat{e}_{\mu}&=e^{\nu}D_{\nu}e_{\mu}\equiv a_{\mu},\\
\dot{e}_{\mu}&=\mathcal{A}u_{\mu}+\alpha_{\mu},\\
\hat{u}_{\mu}&=\left(\frac{1}{3}\Theta+\Sigma\right)e_{\mu}+\Sigma_{\mu}+\varepsilon_{\mu\nu}\Omega^{\nu},\\
\dot{u}_{\mu}&=\mathcal{A}e_{\mu}+\mathcal{A}_{\mu}.
\end{split}
\end{eqnarray*}    
Here the scalar \(\mathcal{A}\) is the acceleration (thought of as the radial component of the acceleration of the unit timelike vector \(u^{\mu}\)). The vector \(a^{\mu}\) is interpreted as the acceleration of \(e^{\mu}\); \(\Sigma\) denotes the scalar associated to the projected, symmetric and trace-free shear tensor \(\sigma_{\mu\nu}\equiv D_{\langle\mu}u_{\nu\rangle}=\Sigma\left(e_{\mu}e_{\nu}-\frac{1}{2}N_{\mu\nu}\right)+2\Sigma_{(\mu}e_{\nu)}+\Sigma_{\mu\nu}\); \(\Theta\equiv D_{\mu}u^{\mu}\) is the expansion; \(\Omega^{\mu}\) is the part of the rotation vector \(\omega^{\mu}=\Omega e^{\mu}+\Omega^{\mu}\) lying in the sheet orthogonal to \(e^{\mu}\); and \(\alpha_{\mu}\equiv \dot{e}_{\bar{\mu}}\).

We also have the following quantities

\begin{eqnarray*}
\begin{split}       
\rho&=T_{\mu\nu}u^{\mu}u^{\nu},\quad\mbox{energy density}\\
p&=\frac{1}{3}h^{\mu\nu}T_{\mu\nu},\quad\mbox{isotropic pressure}\\
\phi&\equiv\delta_{\mu}e^{\mu}\quad\mbox{sheet expansion},\\
\xi&\equiv\frac{1}{2}\varepsilon^{\mu\nu}\delta_{\mu}e_{\nu}\quad\mbox{sheet/spatial twist},\\
\zeta_{\mu\nu}&\equiv\delta_{\langle\mu}e_{\nu\rangle}\quad\mbox{the shear of \(e_{\mu}\)},\\
q_{\mu}&= -h^{\sigma}_{\ \mu}T_{\sigma\nu}u^{\nu}=Qe_{\mu}+Q_{\mu}\quad\mbox{heat flux},\\
\pi_{\mu\nu}&\equiv\Pi\left(e_{\mu}e_{\nu}-\frac{1}{2}N_{\mu\nu}\right)+2\Pi_{(\mu}e_{\nu)}+\Pi_{\mu\nu}.\quad\mbox{anisotropic stress},\\
E_{\mu\nu}&\equiv\mathcal{E}\left(e_{\mu}e_{\nu}-\frac{1}{2}N_{\mu\nu}\right)+2\mathcal{E}_{(\mu}e_{\nu)}+\mathcal{E}_{\mu\nu}.\quad\mbox{electric Weyl},\\
H_{\mu\nu}&\equiv\mathcal{H}\left(e_{\mu}e_{\nu}-\frac{1}{2}N_{\mu\nu}\right)+2\mathcal{H}_{(\mu}e_{\nu)}+\mathcal{H}_{\mu\nu}.\quad\mbox{magnetic Weyl}.
\end{split}
\end{eqnarray*}

The full covariant derivatives of the vectors \(u^{\mu}\) and \(u^{\mu}\) are given respectively by

\begin{subequations}
\begin{align}
\nabla_{\mu}u_{\nu}&=-u_{\mu}\left(Ae_{\nu}+A_{\nu}\right)+\left(\frac{1}{3}\Theta+\Sigma\right)e_{\mu}e_{\nu}+\left(\Sigma_{\nu}+\varepsilon_{\nu\delta}\Omega^{\delta}\right)e_{\mu}+\left(\Sigma_{\mu}-\varepsilon_{\mu\delta}\Omega^{\delta}\right)e_{\nu}+\Omega\varepsilon_{\mu\nu}+\Sigma_{\mu\nu}\nonumber\\
&+\frac{1}{2}\left(\frac{2}{3}\Theta-\Sigma\right)N_{\mu\nu},\label{micel20}\\
\nabla_{\mu}e_{\nu}&=-Au_{\mu}u_{\nu}-u_{\mu}\alpha_{\nu}+\left(\frac{1}{3}\Theta+\Sigma\right)e_{\mu}u_{\nu}+\left(\Sigma_{\mu}-\varepsilon_{\mu\delta}\Omega^{\delta}\right)u_{\nu}+e_{\mu}a_{\nu}+\frac{1}{2}\phi N_{\mu\nu}+\xi\varepsilon_{\mu\nu}+\zeta_{\mu\nu}.\label{mice21}
\end{align}
\end{subequations}

\section{The curvature and some related tensors}\label{soc4}

This section presents the tensors to be utilised in this work as well as equations obtained from contractions of some key identities.

We will not attempt to obtain all of the field equations here. For this, we refer the reader to \cite{cc1} for the detailed obtention of all of the field equations from the \(1+1+2\) splitting. Indeed, the functions \(\lambda\) and \(\beta\) appearing in the form of the Ricci tensor defined below \eqref{pa1} will be written down shortly in terms of the \(1+1+2\) quantities.

As shall be seen, the embedded hypersurfaces to be considered in this work are characterised by just few scalars, and we are interested in how constraints on these scalars determine the local dynamics, as well as the geometry of these hypersurfaces.

In this work we are interested in spacetimes where the Ricci tensor of embedded \(3\)-manifolds (throughout this work embedded \(3\)-manifolds will be horizons from context) assumes the general form

\begin{eqnarray}\label{pa1}
\begin{split}
^3R_{\mu\nu}&=\lambda e_{\mu}e_{\nu}+\beta N_{\mu\nu}\\
&=\left(\lambda-\beta\right)e_{\mu}e_{\nu}+\beta h_{\mu\nu},
\end{split}
\end{eqnarray}
in which case we may write the trace-free part of \eqref{pa1} as

\begin{eqnarray}\label{pak1}
^3S_{\mu\nu}=\left(\lambda-\beta\right)\left(e_{\mu}e_{\nu}-\frac{1}{3}h_{\mu\nu}\right).
\end{eqnarray}
(The 3-dimensional curvatures will be labeled with an upper left superscript 3.) We have the Riemann curvature tensor as

\begin{eqnarray}\label{pa2}
\begin{split}
^3R_{\mu\nu\delta\gamma}&=^3R_{\mu\delta}h_{\nu\gamma}-^3R_{\mu\gamma}h_{\delta\nu}+^3R_{\nu\gamma}h_{\mu\delta}-^3R_{\delta\nu}h_{\mu\gamma}-\frac{^3R}{2}\left(h_{\mu\delta}h_{\nu\gamma}-h_{\mu\gamma}h_{\delta\nu}\right)\\
&=\left(N_{\delta\nu}+e_{\delta}e_{\nu}\right)\biggl[\frac{\lambda}{2}N_{\gamma\mu}+\left(\beta-\frac{\lambda}{2}\right)e_{\gamma}e_{\mu}\biggr] -\left(N_{\gamma\nu}+e_{\gamma}e_{\nu}\right)\biggl[\frac{\lambda}{2}N_{\delta\mu}+\left(\beta-\frac{\lambda}{2}\right)e_{\delta}e_{\mu}\biggr]\\
&+\left(N_{\delta\mu}+e_{\delta}e_{\mu}\right)\left(\beta N_{\gamma\nu}+\lambda e_{\gamma}e_{\nu}\right) -\left(N_{\gamma\mu}+e_{\gamma}e_{\mu}\right)\left(\beta N_{\delta\nu}+\lambda e_{\delta}e_{\nu}\right),
\end{split}
\end{eqnarray}
where \(h_{\mu\nu}\) is the metric induced from the ambient spacetime on the embedded \(3\)-manifold, and \(R\) is the scalar curvature given by

\begin{eqnarray}\label{pa3}
^3R=\lambda+2\beta.
\end{eqnarray}
Let us note here the well known fact that in dimension \(3\), the Weyl tensor vanishes identically. We shall also assume smoothness of the functions \(\lambda\) and \(\beta\), which are functions of the covariant geometric and matter variables.

The Cotton tensor, fully projected to the hypersurface, is given by

\begin{eqnarray}\label{pa4}
\begin{split}
C_{\mu\nu\delta}&=D_{\delta}\left(^3R_{\mu\nu}\right)-D_{\nu}\left(^3R_{\mu\delta}\right)+\frac{1}{4}\left(D_{\nu}\left(^3R\right)h_{\mu\delta}-D_{\delta}\left(^3R\right)h_{\mu\nu}\right)\\
&=2\left(\lambda-\beta\right)\left(e_{\mu}D_{[\delta}e_{\nu]}+e_{[\nu}D_{\delta]}e_{\mu}\right)+2\left[e_{\mu}e_{[\nu}D_{\delta]}+\frac{1}{4}h_{\mu[\delta}D_{\nu]}\right]\lambda+2\left[\left(h_{\mu[\nu}-e_{\mu}e_{[\nu}\right)D_{\delta]}+\frac{1}{2}h_{\mu[\delta}D_{\nu]}\right]\beta\;,
\end{split}
\end{eqnarray}
where we have used the fully orthogonally projected derivative \(D\) (this is the compatible covariant derivative on the hypersurfaces under consideration), and \(C_{\mu\nu\delta}\) is antisymmetric in \(\nu\) and \(\delta\). For dimension \(3\), the cotton tensor can be presented as a tensor density in the form

\begin{eqnarray}\label{pa6}
\begin{split}
C_{\mu}^{\ \nu}&=D_{\delta}\left(^3R_{\sigma\mu}-\frac{1}{4}\left(^3R\right)h_{\sigma\mu}\right)\varepsilon^{\delta\sigma\nu}\\
&=\biggl[\left(\lambda-\beta\right)\left(e_{\sigma}D_{\delta}e_{\mu}+e_{\mu}D_{\delta}e_{\sigma}\right)+\left(e_{\sigma}e_{\mu}-\frac{1}{4}h_{\sigma\mu}\right)D_{\delta}\lambda-\left(e_{\sigma}e_{\mu}-\frac{1}{2}h_{\sigma\mu}\right)D_{\delta}\beta\biggr]\varepsilon^{\delta\sigma\nu}\;,
\end{split}
\end{eqnarray}
sometimes referred to as the Cotton-York tensor.

\subsection{The Ricci tensor for general hypersurfaces in $1+1+2$ spacetimes}

In principle, one may assume any form of the Ricci tensor on a general hypersurface, and study physics on it. However, it may not necessarily be the case that there exists a spacetime in which this hypersurface is embedded. In this section, we provide a minimum set of conditions which are to be satisfied for a hypersurface with Ricci tensor of form \eqref{pa1} to be embedded in a \(1+1+2\) decomposed spacetime.  

Let \(M\) be a spacetime admitting a \(1+1+2\) decomposition, and denote by \(\Xi\subset M\) a codimension \(1\) embedded submanifold of \(M\) (from now onwards hypersurfaces will be denoted by \(\Xi\)). In general, one writes the curvature of \(M\) as (see \cite{gfr1} and references therein)

\begin{eqnarray}\label{3.1}
\begin{split}
R^{\mu\nu}_{\ \ \delta\sigma}&=4u^{[\mu}u_{[\delta}\left(E^{\nu]}_{\ \sigma]}-\frac{1}{2}\pi^{\nu]}_{\ \sigma]}\right)+4\bar{h}^{[\mu}_{\ [\delta}\left(E^{\nu]}_{\ \sigma]}+\pi^{\nu]}_{\ \sigma]}\right)\\
&+2\varepsilon^{\mu\nu\gamma}u_{[\delta}\left(H_{\sigma]\gamma}+\frac{1}{2}\varepsilon_{\sigma]\gamma\varrho}q^{\varrho}\right)+2\varepsilon_{\delta\sigma\gamma}u^{[\mu}\left(H^{\nu\gamma}+\frac{1}{2}\varepsilon^{\nu]\gamma\varrho}q_{\varrho}\right)+\frac{2}{3}\left(\rho+3p-2\Lambda\right)u^{[\mu}u_{[\delta}\bar{h}^{\nu]}_{\ \sigma]}+\frac{2}{3}\left(\rho+\Lambda\right) \bar{h}^{\mu}_{\ [\delta}\bar{h}^{\nu}_{\ \sigma]},
\end{split}
\end{eqnarray}
where \(\Lambda\) is the cosmological constant and \(\bar{h}^{\mu\nu}=g^{\mu\nu}+u^{\mu}u^{\nu}\). One may write \eqref{3.1} in its full \(1+1+2\) form by decomposing \(E_{\mu\nu},H_{\mu\nu}\) and \(\pi_{\mu\nu}\) using \eqref{micel6}, and \(q_{\mu}\) using \eqref{micel4}. We define the curvature quantities on \(\Xi\) as follows \cite{ge2}: define a normal to \(\Xi\) as

\begin{eqnarray}\label{3.2}
n^{\mu}=a_1u^{\mu}+a_2e^{\mu}.
\end{eqnarray}
(Whether \(n^{\mu}\) is spacelike of timelike willmplace some constraints on \(a_1\) and \(a_2\).) The first fundamental form of \(\Xi\) is given by

\begin{eqnarray}\label{3.3}
h_{\mu\nu}=g_{\mu\nu}\mp n_{\mu}n_{\nu}=N_{ab}-\left(1\pm a_1^2\right)u_{\mu}u_{\nu}+\left(1\mp a_2^2\right)e_{\mu}e_{\nu}\mp 2a_1a_2u_{(\mu}e_{\nu)},
\end{eqnarray}
where the choice of the  ``-" or ``+" sign depends of the whether \(\Xi\) is \textit{timelike} or \textit{spacelike} respectively. Note the relationship between the tensors \(\bar{h}_{\mu\nu}\) and \(h_{\mu\nu}\) here:

\begin{eqnarray}\label{yes}
\bar{h}_{\mu\nu}=h_{\mu\nu}\Big|_{a_1=\pm 1;\ a_2=0}.
\end{eqnarray}
(In this case it is the ``-" sign that is chosen for \(\left(1\pm a_1^2\right)\).) We also note that

\begin{eqnarray*}
\left(E_{\mu\nu},H_{\mu\nu},\pi_{\mu\nu}\right)u^{\mu}=\left(E_{\mu\nu},H_{\mu\nu},\pi_{\mu\nu}\right)\bar{h}^{\mu\nu}=0.
\end{eqnarray*}
The second fundamental form is then calculated as

\begin{eqnarray}\label{3.4}
\begin{split}
\chi_{\mu\nu}&=h^{\delta}_{\ (\mu}h^{\sigma}_{\ \nu)}\nabla_{\delta}n_{\sigma}\\
&=Z_1u_{\mu}u_{\nu}-Z_2e_{\mu}e_{\nu}\mp Z_3u_{(\mu}e_{\nu)}+\frac{1}{2}\left(a_1\left(\frac{2}{3}\Theta-\Sigma\right)+a_2\phi\right)N_{\mu\nu}\\
&+a_1\Sigma_{\mu\nu}+a_2\zeta_{\mu\nu}+\left(1\pm a_1^2\right)Z'_{\mu\nu}+\left(1\mp a_2^2\right)\left[a_1Z_{\mu\nu}+a_2\bar{Z}_{\mu\nu}\right],
\end{split}
\end{eqnarray}
where we have defined the quantities

\begin{subequations}
\begin{align}
Z_1&=\left(1\pm a_1^2\right)\left(a_1a_2\dot{a}_2-a_1a_2\hat{a}_1-a_2A-\left(1\pm a_1^2\right)\dot{a}_1\right)+a_1a_2^2\left(a_1\hat{a}_2-\left(\frac{1}{3}\Theta+\Sigma\right)\right),\label{3.20a}\\
Z_2&=\left(1\mp a_2^2\right)\left(a_1a_2\dot{a}_2-a_1a_2\hat{a}_1-a_1\left(\frac{1}{3}\Theta+\Sigma\right)-\left(1\mp a_2^2\right)\hat{a}_2\right)+a_1^2a_2\left(a_2\dot{a}_1+A\right),\label{3.21b}\\
Z_3&=\left(1\pm a_1^2\right)\left(\left(1\mp a_2^2\right)\left(\hat{a}_1-\dot{a}_2\right)-2a_1a_2\dot{a}_1-a_1A\right)+\left(1\mp a_2^2\right)\biggl(a_2\left(\frac{1}{3}\Theta+\Sigma\right)-2a_1a_2\hat{a}_2\biggr)\nonumber\\
&+a_1^2a_2\left(a_2\left(\dot{a}_2-\hat{a}_1\right)-\left(\frac{1}{3}\Theta+\Sigma\right)\right),\label{3.22}\\
Z_{\mu\nu}&=2\left(\Sigma_{(\nu}-\varepsilon_{\nu\delta(\nu}\Omega^{\delta}+a_2a_{(\nu}\right)e_{\mu)},\label{3.23}\\
\bar{Z}_{\mu\nu}&=2\left(\Sigma_{(\mu}+\varepsilon_{\delta(\mu}\Omega^c+\delta_{(\mu}a_2\right)e_{\nu)},\label{3.24}\\
Z'_{\mu\nu}&=2\left(u_{(\nu}\delta_{\mu)}a_1-u_{(\mu}\alpha_{\nu)}-a_1u_{(\mu}A_{\nu)}\right).\label{3.25}
\end{align}
\end{subequations}
From \eqref{3.1}, we have the Ricci curvature on \(M\) as

\begin{eqnarray}\label{3.5}
R_{ab}=\pi_{ab}+2q_{(a}u_{b)}-\frac{1}{2}\left(\rho+3p-2\Lambda\right)\left(\frac{1}{3}\bar{h}_{ab}-u_au_b\right)+\frac{2}{3}\left(\rho+\Lambda\right)\bar{h}_{ab}.
\end{eqnarray}
The scalar curvature is therefore

\begin{eqnarray}\label{3.6}
R=\rho-3p+4\Lambda.
\end{eqnarray}
The relationship between the curvature of \(\Xi\), \(^3R^{\mu}_{\ \nu\delta\sigma}\), and the curvature of \(M\), \(R^{\mu}_{\ \nu\delta\sigma}\), is given by \cite{ge2}

\begin{eqnarray}\label{3.7}
^3R^{\mu}_{\ \nu\delta\sigma}=R^{\bar{\mu}}_{\ \bar{\nu}\bar{\delta}\bar{\sigma}}h^{\mu}_{\ \bar{\mu}}h^{\bar{\nu}}_{\ 
\nu}h^{\bar{\delta}}_{\ \delta}h^{\bar{\sigma}}_{\ \sigma}\pm \chi^{\mu}_{\ \delta}\chi_{\nu\sigma} \mp \chi^{\mu}_{\ \sigma}\chi_{\nu\delta},
\end{eqnarray}
where, from now on, we are denoting all curvature quantities associated to \(\Xi\) with an overhead `\textit{tilde}'. From \eqref{3.6} we obtain

\begin{eqnarray}\label{3.8}
^3R_{\mu\nu}&=\underline{R}_{\mu\nu}\mp \bar{R}_{\mu\nu}\pm\chi\left(\chi_{\mu\nu}\right) \mp \underline{\chi}_{\mu\nu},
\end{eqnarray}
with the associated scalar curvature is

\begin{eqnarray}\label{3.9}
^3R&=\left(R\pm\chi^2\right)\mp\left(R_*+\bar{\chi}\right).
\end{eqnarray}
where we have defined

\begin{subequations}
\begin{align}
\chi&=\chi^{\mu}_{\ \mu}=-Z_1-Z_2+a_1\left(\frac{2}{3}\Theta-\Sigma\right)+a_2\phi,\label{3.10}\\
\bar{\chi}&=\chi_{\mu\nu}\chi^{\mu\nu}=Z_1^2+Z^2-Z_3^2+a_1^2\left(\frac{1}{2}\left(\frac{2}{3}\Theta-\Sigma\right)^2+2\Omega^2\right)+a_2^2\left(\frac{1}{2}\phi^2+2\xi^2\right)+2a_1^2\left(a_{\mu}\Sigma^{\mu}+\varepsilon_{\mu\nu}\left(\Sigma^{\mu}+a_{\mu}a^{\mu}\right)\Omega^{\nu}\right)\nonumber\\
&+\left(1\mp a_2^2\right)^2\biggl[\left(a_1^2+a_2^2\right)\left(\Sigma_{\mu}\Sigma^{\mu}+\Omega_{\mu}\Omega^{\mu}\right)+a_2^2\biggl(a_1^2a_{\mu}a^{\mu}+\delta_{\mu}a_2\delta^{\mu}a_2\biggr)+2a_2^2\left(\delta_{\mu}a_2\Sigma^{\mu}-\varepsilon_{\mu\nu}\left(\Sigma^{\mu}+\delta^{\mu}a_2\right)\Omega^{\nu}\right)\biggr]\nonumber\\
&+2a_1a_2\left(\frac{1}{2}\phi\left(\frac{2}{3}\Theta-\Sigma\right)+2\Omega\xi\right)-\left(1\pm a_1^2\right)^2\biggl(\delta_{\mu}a_1\delta^{\mu}a_1+\alpha_{\mu}\alpha^{\mu}+a_1^2A_{\mu}A^{\mu}+2a_1\alpha_{\mu}A^{\mu}\biggr)+a_1^2\Sigma_{\mu\nu}\Sigma^{\mu\nu}\nonumber\\
&+a_2^2\zeta_{\mu\nu}\zeta^{\mu\nu},\label{3.11}\\
\underline{\chi}_{\mu\nu}&=\chi^{\delta}_{\ \mu}\chi_{\delta\nu}=\left[-Z_1^2\mp Z_3^2+\left(1\pm a_1^2\right)^2\delta_aa_1\delta^aa_1\right]u_{\mu}u_{\nu}+\biggl[Z_2^2\pm\frac{1}{4}Z_3^2+a_2^2\left(1\pm a_2^2\right)^2V\biggr]e_{\mu}e_{\nu}\nonumber\\
&+2\left(1\pm a_1^2\right)\left(\alpha_{(\nu}+a_1A_{(\nu}\right)\left(Z_1u_{\mu)}+\frac{1}{2}Z_3e_{\mu)}\right)+2\left(1\pm a_1^2\right)\left(a_1\Sigma_{\delta(\mu}+a_2\zeta_{\delta(\mu}\right)u_{\nu)}\delta^{\delta}a_1\nonumber\\
&+2\left[\pm\frac{1}{2}Z_3\left(Z_1+Z_2\right)+a_2\left(1\pm a_1^2\right)\left(1\mp a_2^2\right)\delta_{\delta}a_1\left(\Sigma^{\delta}-\varepsilon^{\delta\sigma}\Omega_{\sigma}+\delta^{\delta}a_2\right)\right]u_{(\mu}e_{\nu)}\nonumber\\
&+\left(1\pm a_1^2\right)^2\left(a_1^2A_{\mu}A_{\nu}-\alpha_{\mu}\alpha_{\nu}-2a_1\alpha_{(\mu}A_{\nu)}\right)+a_1\left(1\mp a_2^2\right)\left(\Sigma_{\nu}-\varepsilon_{\delta(\nu}\Omega^{\delta}+a_2a_{(\nu}\right)\biggl(\frac{1}{2}Z_3u_{\mu)}-Z_2e_{\mu)}\biggr)\nonumber\\
&+\frac{1}{2}\left(a_1\left(\frac{2}{3}\Theta-\Sigma\right)+a_2\phi\right)V_{\mu\nu}+2a_2^2\left(1\mp a_2^2\right)\left(\Sigma^{\delta}-\varepsilon^{\delta\sigma}\Omega_{\sigma}+\delta^{\delta}a_2\right)\zeta_{\delta(\mu}e_{\nu)}+\left(a_1\Omega+a_2\xi\right)\bar{V}_{\mu\nu}\nonumber\\
&+2a_1a_2\left[\zeta^{\delta}_{(\mu}\Sigma_{\nu)\delta}+\left(1\mp a_2^2\right)\left(\Sigma^{\delta}-\varepsilon^{\delta\sigma}\Omega_{\sigma}+\delta^{\delta}a_2\right)\Sigma_{\delta(\mu}e_{\nu)}\right],\label{3.12}\\
\underline{R}_{\mu\nu}&=R_{\bar{\nu}\bar{\sigma}}h^{\bar{\nu}}_{\ \mu}h^{\bar{\sigma}}_{\ \nu}=\frac{1}{2}\left(\rho-p-\Pi+2\Lambda\right)N_{\mu\nu}+\frac{1}{2}\biggl[\left(1\pm a_1^2\right)\biggl(\left(1\pm a_1^2\right)\left(\rho+3p-2\Lambda\right)\mp a_1a_2Q\biggr)\nonumber\\
&\mp a_1a_2\left(\left(1\pm a_1^2\right)Q\mp a_1a_2\left(\rho-p+2\Pi+2\Lambda\right)\right)\biggr]u_{\mu}u_{\nu}+\left(Q_{(\mu}+2\Pi_{(\mu}\right)\left[\mp a_1a_2u_{\nu)}+\left(1\mp a_2^2\right)e_{\nu)}\right]\nonumber\\
&+2\left[\left(1\mp a_2^2\right)u_{(\mu}\mp a_1a_2 e_{(\mu}\right]\Pi_{\nu)}+\frac{1}{2}\biggl[\left(1\mp a_2^2\right)\left(\left(1\mp a_2^2\right)\left(\rho-p+2\Pi+2\Lambda\right)\pm a_1a_2Q\right)\nonumber\\
&+2\left[\left(1\pm a_1^2\right)u_{(\mu}\pm a_1a_2 e_{(\mu}\right]Q_{\nu)}\pm a_1a_2\left(\left(1\mp a_2^2\right)Q\pm a_1a_2\left(\rho+3p-2\Lambda\right)\right)\biggr]e_{\mu}e_{\nu}\nonumber\\
&+\biggl[\left(1\mp a_2^2\right)\biggl(\left(1\pm a_1^2\right)Q\mp a_1a_2\left(\rho-p+2\Pi+2\Lambda\right)\biggr)\pm a_1a_2\left(\left(1\pm a_1^2\right)\left(\rho+3p-2\Lambda\right)\mp a_1a_2Q\right)\biggr]u_{(\mu}e_{\nu)},\label{3.13}\\
\bar{R}_{\mu\nu}&=R^{\bar{\mu}}_{\ \bar{\nu}\bar{\delta}\bar{\sigma}}n_{\bar{\mu}}n^{\bar{\delta}}h^{\bar{\nu}}_{\ 
\nu}h^{\bar{\sigma}}_{\ \mu}=\frac{1}{2}\biggl[\left(a_1a_2Q+\left(a_2^2-1\right)\mathcal{E}+\frac{1}{2}\left(a_2^2+1\right)\Pi\right)+\frac{1}{3}\left(\rho+3p-2\Lambda\right)\biggr]N_{\mu\nu}\nonumber\\
&+a_2^2\left[\left(\mathcal{E}-\frac{1}{2}\Pi\right)+\frac{1}{6}\left(\rho+3p-2\Lambda\right)\right]u_{\mu}u_{\nu}+a_1^2\biggl[\left(\mathcal{E}-\frac{1}{2}\Pi\right)+\frac{1}{6}\left(\rho+3p-2\Lambda\right)\biggr]e_{\mu}e_{\nu}\nonumber\\
&+2a_1a_2\left[\left(\mathcal{E}-\frac{1}{2}\Pi\right)+\frac{1}{6}\left(\rho+3p-2\Lambda\right)\right]u_{(\mu}e_{\nu)}+\frac{1}{4}a_2^2\biggl(\mathcal{E}_{\mu\nu}+\frac{1}{2}\Pi_{\mu\nu}\biggr)\nonumber\\
&+2\left[a_1a_2\left(\mathcal{E}_{(\nu}-\frac{1}{2}\Pi_{(\nu}\right)-\frac{1}{2}a_2^2\left(1\pm a_1^2\right)\left(\mathcal{H}^{\delta}\varepsilon_{\delta(\nu}-\frac{1}{2}Q_{(\nu}\right)\mp a_1^2a_2^2\mathcal{H}^{\delta}\varepsilon_{\delta(\nu}\right]u_{\mu)}\nonumber\\
&+2\left[a_1^2\left(\mathcal{E}_{(\nu}-\frac{1}{2}\Pi_{(\nu}\right)\mp a_1a_2^3\left(\mathcal{H}^{\delta}\varepsilon_{\delta(\nu}-\frac{1}{2}Q_{(\nu}\right) + a_1a_2\left(1\mp a_2^2\right)\mathcal{H}^{\delta}\varepsilon_{\delta(\nu}\right]e_{\mu)}\nonumber\\
&+2\left[a_1a_2\left(\mathcal{E}_{(\mu}-\frac{1}{2}\Pi_{(\mu}\right)+\frac{1}{2}a_2^2\left(1+a_1^2\pm a_1^2\right)\left(\mathcal{H}^{\delta}\varepsilon_{\delta(\mu}-\frac{1}{2}Q_{(\mu}\right)\right]u_{\nu)}\nonumber\\
&+2\left[a_1^2\left(\mathcal{E}_{(\mu}-\frac{1}{2}\Pi_{(\mu}\right)+\frac{1}{2}a_1a_2\left(1-a_2^2\mp a_2^2\right)\left(\mathcal{H}^{\delta}\varepsilon_{\delta(\mu}-\frac{1}{2}Q_{(\mu}\right)\right]e_{\nu)},\label{3.14}\\
R_*&=R_{\mu\nu}n^{\mu}n^{\nu}=\frac{1}{2}a_2^2\left(\rho+p+2\Lambda\right)-\frac{1}{2}a_1^2\left(\rho+3p-2\Lambda\right)+\frac{1}{2}a_2Q\left(3a_2-a_1\right)+a_2^2\Pi+a_1^2\left(\mathcal{E}_{\mu\nu}-\frac{1}{2}\Pi_{\mu\nu}\right),\label{3.15}
\end{align}
\end{subequations}
with

\begin{subequations}
\begin{align}
V&=\Sigma_{\mu}\Sigma^{\mu}+\Omega_{\mu}\Omega^{\mu}+\delta_{\mu}a_1\delta^{\mu}a_1-2\Sigma_{\mu}\varepsilon^{\mu\nu}\Omega_{\nu},\label{3.16}\\
V_{\mu\nu}&=\frac{1}{2}\left(a_1\left(\frac{2}{3}\Theta-\Sigma\right)+a_2\phi\right)N_{\mu\nu}+2a_2\left(1\mp a_2^2\right)\left(\Sigma_{(\mu}+\Sigma_{\delta(\mu}\Omega^{\delta}+\delta_{(\mu}\right)e_{\nu)}+2\left(a_1\Sigma_{\mu\nu}+a_2\zeta_{\mu\nu}\right),\label{3.17}\\
\bar{V}_{\mu\nu}&=2\left(a_1\Omega+a_2\xi\right)N_{\mu\nu}+2\varepsilon_{\delta(\mu}\biggl[a_2\left(1\mp a_2^2\right)\left(\Sigma^{\delta}-\varepsilon^{\delta\sigma}\Omega_{\sigma}+\delta^{\delta}a_2\right)e_{\nu)}+\left(1\pm a_1^2\right)u_{\nu)}\delta^{\delta}a_1+a_1\Sigma_{\nu)}^{\ \delta}+a_2\zeta_{\nu)}^{\ \delta}\biggr].\label{3.18}
\end{align}
\end{subequations}

Now, notice first of all that \eqref{pa1} has no mixed term or \(u_{\mu}u_{\nu}\) term. Also, there are no terms constructed from product of \(2\)-vectors or tensor quantities. The choice of the first fundamental form is just \(\bar{h}_{\mu\nu}=g_{\mu\nu}+u_{\mu}u_{\nu}\), so that \(a_2=0\) and \(a_1=1\) in which case \(h^{\mu\nu}\) and \(\bar{h}^{\mu\nu}\) coincide. In this case, it is easy to see that \(Z_1=Z_3=0\) and \(Z_2=-\left((1/3)\Theta+\Sigma\right)\). The \(u_{\mu}u_{\nu}\) and mixed terms are identically zero. In addition, the following condition on the hypersurface is required for \eqref{3.8} to reduce to \eqref{pa1}:

\begin{eqnarray}\label{3.19}
\begin{split}
0&=4\mathcal{E}_{(\mu}u_{\nu)}+\left(Q_{(\mu}+2\Pi_{\mu}\right)e_{\nu)}+\left(\frac{2}{3}\Theta-\Sigma\right)\Sigma_{\mu\nu}+2\Omega\varepsilon_{\delta(\mu}\Sigma_{\nu)}^{\ \delta}\\
&-\Theta\left[\Sigma_{\mu\nu}+2\left(\Sigma_{(\mu}-\varepsilon_{\delta(\mu}\Omega^{\delta}\right)e_{\nu)}\right].
\end{split}
\end{eqnarray}

The scalars \(\lambda\) and \(\beta\) in \eqref{pa1} can now be explicitly written as

\begin{subequations}
\begin{align}
\lambda&=\frac{2}{3}\left(\rho+\Lambda\right)+\mathcal{E}+\frac{1}{2}\Pi-\left(\frac{1}{3}\Theta+\Sigma\right)\left(\frac{2}{3}\Theta-\Sigma\right),\label{3.20}\\
\beta&=\frac{2}{3}\rho-\frac{1}{2}\left(\mathcal{E}+\frac{1}{2}\Pi\right)-\frac{1}{2}\left(\frac{2}{3}\Theta-\Sigma\right)\left(\frac{2}{3}\Theta+\frac{1}{2}\Sigma\right)+2\Omega^2.\label{3.21}
\end{align}
\end{subequations}
Perhaps a well known class of spacetimes having spacelike hypersurfaces with Ricci tensor of such is the LRS II class (see \cite{gbc1}). These are the observers' rest spaces which play a fundamental role in obtaining exact solutions to the Einsten's field equations.

\section{Characterization, the equations and constraints}\label{soc5}

In this section we provide a characterization of locally symmetric hypersurfaces in spacetimes admitting a \(1+1+2\) decomposition, with Ricci tensor of the form \eqref{pa1}. We then obtain additional equations and constraints that aid further analysis of, and restrictions on these hypersurfaces.

The condition of local symmetry of Riemannian manifolds is given by the vanishing of the first covariant derivative of the curvature tensor:

\begin{eqnarray}\label{panel1}
D_{\sigma}\left(^3R_{\mu\nu\delta\gamma}\right)=0.
\end{eqnarray}
Properties of locally symmetric Riemannian spaces \cite{al1,kn1,st1} is a well studied subject and complete classification schemes have been provided. The Lorentzian cases have been studied as well \cite{st1,hda1,cdc1}, with classifications provided up to the \(2\)-symmetric and semi-symmetric cases relatively recently by Senovilla \cite{js1}. Riemannian manifolds satisfying \eqref{panel1} have been shown to satisfy the bi implication

\begin{eqnarray}\label{panel2}
D_{\sigma_1}\ldots D_{\sigma_k}\left(^3R_{\mu\nu\delta\gamma}\right)=0\iff D_{\sigma}\left(^3R_{\mu\nu\delta\gamma}\right)=0,
\end{eqnarray}
where \(D_{\sigma_k}\) denotes the \(k^{th}\) covariant derivative. Hence, \(C_{\mu\nu\delta}=0\) in which case the hypersurfaces we are considering here are conformally flat, and therefore both \(\mathcal{E}\) and \(\mathcal{H}\) are zero (of course this is well known). By contracting \eqref{pa4} with \(e^{\delta}\varepsilon^{\mu\nu},e^{\delta}N^{\mu\nu}\) and \(e^{\delta}e^{\mu}e^{\nu}\) gives respectively

\begin{subequations}
\begin{align}
0&=\left(\lambda-\beta\right)\xi,\label{panew1}\\
0&=\left(\lambda-\beta\right)\phi+\frac{1}{2}\left(\hat{\lambda}-2\hat{\beta}\right),\label{panew2}\\
0&=\left(\hat{\lambda}-2\hat{\beta}\right),\label{panew3}
\end{align}
\end{subequations}
which reduces to the set

\begin{subequations}
\begin{align}
0&=\left(\lambda-\beta\right)\xi,\label{pane1}\\
0&=\left(\lambda-\beta\right)\phi,\label{pane2}
\end{align}
\end{subequations}
Hence, we shall consider the following configurations:

\begin{subequations}
\begin{align}
\lambda-\beta&=0;\qquad\mbox{(Einstein manifold)}\qquad\mbox{or}\label{pane3}\\
\phi=\xi&=0\qquad\left(\mbox{assuming}\quad\lambda-\beta\neq 0\right),\label{pane4}
\end{align}
\end{subequations}
where the condition \(\lambda-\beta=0\) is simply

\begin{eqnarray}\label{ddx}
\frac{2}{3}\Lambda+\frac{3}{4}\Pi-\frac{3}{4}\left(\frac{2}{3}\Theta-\Sigma\right)-2\Omega^2=0.
\end{eqnarray}
We therefore have that the set
\begin{eqnarray}\label{der}
\mathcal{D}:=\lbrace{\lambda,\beta,\xi,\phi\rbrace},
\end{eqnarray}
characterizes embedded locally symmetric \(3\)-manifolds in \(4\)-dimensional spacetimes with Ricci tensor of the form \eqref{pa1}. It is also clear that the considered hypersurfaces are flat only in the Einstein case with \(\lambda=\beta=0\). This allows us to state the first useful proposition:

\begin{proposition}\label{prop1}
Let \(\left(M,g_{\mu\nu}\right)\) be a \(4\)-dimensional spacetime and let \(\Xi\) be a locally symmetric embedded \(3\)-manifold with induced metric \(h_{\mu\nu}\), and with Ricci tensor of the form \eqref{pa1}. Then \(\Xi\) is either

\begin{enumerate}

\item an Einstein space; or

\item is non-twisting with vanishing sheet expansion.

\end{enumerate}
\end{proposition}

Both cases of Proposition \ref{prop1} are very important cases that will be given the necessary considerations. For example, Ricci solitons, which are steady solutions of the Ricci flow evolution equation can be seen as `\textit{perturbations}' of Einstein spaces. On the other hand, the class of \textit{locally rotationally symmetric spacetimes}, which contains a lot of physically relevant spherically symmetric spacetimes in general relativity, have important subclasses that are non-twisting. In context of marginally trapped tubes however (these generalize the boundaries of black holes), the case \(2.\) of Proposition \ref{prop1} are minimal.

Now, using the property of the vanishing of the divergence of \eqref{pa6}, we contract \eqref{pa6} by \(e^{\eta}h_{\eta}^{\mu}D_{\nu}\) and obtain 

\begin{eqnarray}\label{pane7}
2\xi\hat{\beta}+J\left(\lambda-\beta\right)=0,
\end{eqnarray}
where we have defined the operator \(J=a_{\mu}\varepsilon^{\mu\nu}\delta_{\nu}\).

The contracted  Ricci identities, \(D^{\mu}R_{\mu\nu}=0\), obtained from \eqref{panel1} can be expressed as

\begin{eqnarray}\label{panel4}
\left(\lambda-\beta\right)\left(a_{\nu}+\phi e_{\nu}\right)+e_{\nu}\hat{\lambda}+\delta_{\nu}\beta=0,
\end{eqnarray}
which upon contracting with \(e^{\nu}\) and \(u^{\nu}\) we obtain respectively

\begin{eqnarray}\label{panel5}
\left(\lambda-\beta\right)\phi+\hat{\lambda}&=0.
\end{eqnarray}
In any case the first term on the left hand side will vanish by Proposition \ref{prop1}, and hence we must have that \(\hat{\lambda}=0\) which would imply that \(\hat{\beta}=0\) by \eqref{panew3}. In this case we have that \eqref{pane7} reduces to 

\begin{eqnarray}\label{fray}
J\left(\lambda-\beta\right)=0.
\end{eqnarray}

Since \(u^{\mu}R_{\mu\nu\delta\gamma}=0\), the tensor \(D_{\mu}D_{\nu}u_{\delta}\) is symmetric in the \(\mu\) and \(\nu\) indices by the Ricci identities for \(u^{\mu}\). We note that

\begin{eqnarray}
\begin{split}
D_{\mu}u_{\nu}&=\left(\frac{1}{3}\Theta+\Sigma\right)e_{\mu}e_{\nu}+\frac{1}{2}\left(\frac{2}{3}\Theta-\Sigma\right)N_{\mu\nu}+\left(\Sigma_{\nu}+\varepsilon_{\nu\sigma}\Omega^{\sigma}\right)e_{\mu}+\left(\Sigma_{\mu}-\varepsilon_{\mu\sigma}\Omega^{\sigma}\right)e_{\nu}+\Omega\varepsilon_{\mu\nu}+\Sigma_{\mu\nu}.
\end{split}
\end{eqnarray}
Hence, contracting \(D_{\mu}D_{\nu}u_{\delta}\) with \(\varepsilon^{\mu\nu}\) gives zero. Explicitly, we write this as the equation

\begin{eqnarray}\label{panel10}
\begin{split}
0&=\left[3\Sigma\xi+\varepsilon^{\nu\sigma}D_{\nu}\left(\Sigma_{\sigma}-\varepsilon_{\sigma\delta}\Omega^{\delta}\right)\right]e_{\mu}+2\xi\left(\Sigma_{\mu}+\varepsilon_{\mu\sigma}\Omega^{\sigma}\right)+\frac{1}{2}\varepsilon^{\sigma}_{\mu}\delta_{\sigma}\left(\frac{2}{3}\Theta-\Sigma\right)\\
&+\varepsilon^{\sigma\nu}\left(\Omega D_{\sigma}\varepsilon_{\nu\mu}+D_{\sigma}\Sigma_{\nu\mu}\right)+\delta_{\mu}\Omega+\left(\Sigma_{\nu}-\varepsilon_{\nu\sigma}\Omega^{\sigma}\right)\left(\frac{1}{2}\phi\varepsilon^{\nu}_{\mu}-\xi N^{\nu}_{\mu}+\varepsilon^{\sigma\nu}\zeta_{\sigma\mu}\right),
\end{split}
\end{eqnarray}
and upon contracting \eqref{panel10} with \(e^{\mu}\) and simplifying gives the following expression:

\begin{eqnarray}\label{peen1}
3\Sigma\xi+\Omega \phi-\delta_{\mu}\Omega^{\mu}+\varepsilon^{\mu\nu}\left[\delta_{\mu}\Sigma_{\nu}+\Sigma^{\sigma}_{\nu}\left(\delta_{\mu}e_{\sigma}+\delta_{\{\mu}e_{\sigma\}}\right)\right]&=0.
\end{eqnarray}

Now, the Gauss and Codazzi embedding equations to be satisfied by a properly embedded hypersurface, with Ricci tensor of the form \eqref{pa1}, in the ambient spacetime are explicitly given by

\begin{subequations}
\begin{align}
0&=\left[\left(\lambda-\beta\right)+\frac{3}{2}\left(\frac{1}{3}\Theta\Sigma-\mathcal{E}-\frac{1}{2}\Pi\right)\right]\left(e_{\mu}e_{\nu}-\frac{1}{3}h_{\mu\nu}\right)+\frac{1}{3}\Theta\Sigma_{\mu\nu}+\mathcal{E}_{\mu\nu}\nonumber\\
&+2\left(\frac{1}{3}\Theta\Sigma_{(\mu}+\mathcal{E}_{(\mu}+\frac{1}{2}\Pi_{(\mu}\right)e_{\nu)}+\frac{1}{3}h_{\mu\nu}\left(\frac{3}{2}\Sigma^2+2\Sigma_{\mu}\Sigma^{\mu}+\Sigma_{\mu\nu}\Sigma^{\mu\nu}\right),\label{new1}\\
0&=\lambda+2\beta+\frac{2}{3}\Theta^2-\left(\frac{3}{2}\Sigma^2+2\Sigma_{\mu}\Sigma^{\mu}+\Sigma_{\mu\nu}\Sigma^{\mu\nu}\right)-2\rho,\label{new2}\\
0&=-\frac{3}{2}e_{(\mu}\varepsilon_{\ \nu)}^{\sigma}\delta_{\sigma}\Sigma-\frac{3}{2}\left(\xi\Sigma-\mathcal{H}\right) \left(h_{\mu\nu}-e_{\mu}e_{\nu}\right)+\xi\Sigma_{(\mu}e_{\nu)}+\varepsilon_{\sigma\delta(\mu}D^{\sigma}\Sigma_{\nu)}^{\delta}\nonumber\\
&-\varepsilon_{\sigma(\mu}D^{\sigma}\Sigma_{\nu)}+\Sigma_{\sigma}\left(\varepsilon^{\ \sigma}_{(\mu}a_{\nu)}+\frac{1}{2}\phi e_{(\mu}\varepsilon_{\nu)}^{\ \sigma}\right)+\frac{1}{2}\left(\Sigma_{\delta}\varepsilon^{\ \delta}_{\sigma(\mu}-3\varepsilon_{\sigma(\mu}\right)\zeta^{\ \sigma}_{\nu)}\nonumber\\
&+\left(\varepsilon_{\sigma\delta(\mu}D^{\sigma}\Sigma^{\delta}+\mathcal{H}_{(\mu}\right)e_{\nu)}+e_{(\nu}\Sigma_{\mu)},\label{new3}\\
0&=\left[\frac{3}{2}\Sigma\phi+\frac{3}{4}\hat{\Sigma}-\frac{2}{3}\hat{\Theta}-Q+\left(\delta_{\mu}-2a_{\mu}\right)\Sigma^{\mu}\right]e_{\mu}+\frac{3}{2}\left(\Sigma a_{\mu}+\phi\Sigma_{\mu}\right)\nonumber\\
&+\hat{\Sigma}_{\mu}-Q_{\mu}+\left(\xi\varepsilon_{\sigma\mu}+\zeta_{\sigma\mu}\right)\Sigma^{\sigma}-\frac{3}{4}\delta_{\mu}\Sigma-\frac{2}{3}\delta_{\mu}\Theta.\label{new4}
\end{align}
\end{subequations}
Let us take the trace of \eqref{new1} and \eqref{new3} as well as contract with \(e^{\mu}e^{\nu}\). We obtain the set 

\begin{subequations}
\begin{align}
0&=\frac{3}{2}\Sigma^2+2\Sigma_{\mu}\Sigma^{\mu}+\Sigma_{\mu\nu}\Sigma^{\mu\nu},\label{exe1}\\
0&=3\left(\mathcal{H}-\xi\Sigma\right)-\varepsilon^{\mu\nu}\left(\delta_{\mu}-a_{\mu}\right)\Sigma_{\nu},\label{exe2}\\
0&=\frac{2}{3}\left[\left(\lambda-\beta\right)+\frac{1}{2}\left(\frac{1}{3}\Theta\Sigma-\frac{1}{2}\Pi\right)\right]+\frac{1}{3}\left[\frac{3}{2}\Sigma^2+2\Sigma_{\mu}\Sigma^{\mu}+\Sigma_{\mu\nu}\Sigma^{\mu\nu}\right],\label{exe3}\\
0&=\Sigma_{\mu\sigma}\varepsilon^{\mu\nu}\delta_{\nu}e^{\sigma},\label{exe4}
\end{align}
\end{subequations}
(keep in mind that \(\mathcal{H}\) is zero everywhere since the hypersurface is conformally flat) and upon contracting \eqref{new4} with \(e^{\mu}\) we get

\begin{eqnarray}\label{exe5}
0=\frac{3}{2}\Sigma\phi+\frac{3}{4}\hat{\Sigma}-\frac{2}{3}\hat{\Theta}-Q+\left(\delta_{\mu}-2a_{\mu}\right)\Sigma^{\mu}.
\end{eqnarray}
Comparing \eqref{exe1} to \eqref{new2}, we have that

\begin{eqnarray}\label{exe6}
\frac{2}{3}\Theta^2=2\rho-\left(\lambda+2\beta\right),
\end{eqnarray}
and hence, 

\begin{eqnarray}\label{exe7}
\lambda+2\beta=R\leq2\rho.
\end{eqnarray}
Indeed, for finite \(\rho\), the scalar curvature is finite for non-negative \(R\). This would then indicate compactness of the hypersurface since the scalar curvature is assumed to be bounded below. Furthermore, this implies that the energy density \(\rho\) is non-negative and \(\rho=0\) which implies \(^3R=0\). 

We can use \eqref{3.20} and \eqref{3.21} to substitute into \eqref{exe6} and show that the cosmological constant is proportional to the square of the shear scalar \(\Sigma\), and must be negative in the cases considered in this work. In particular we have that \(\Lambda=-(9/4)\Sigma^2\). Therefore, the cosmological constant would vanish if and only if the hypersurface is not shearing. Clearly, requiring \(R\) to be non-negative forces the energy density \(\rho\) to be non-negative on the hypersurfaces as well.

Given a spacetime (of type considered in this work), choosing a hypersurface in the spacetime (this implies specifying the Ricci tensor on the hypersurface, which in turn implies specifying \(\lambda\) and \(\beta\)) will present additionals constraint on the hypersurfaces. So, for example, consider the class II of locall rotationally symmetric spacetimes, with \cite{gbc1}

\begin{eqnarray}
\begin{split}
\lambda&=-\left(\hat{\phi}+\frac{1}{2}\phi^2\right),\label{exe8}\\
\beta&=-\left[\frac{1}{2}\left(\hat{\phi}+\phi^2\right)-K\right],\label{exe9}
\end{split}
\end{eqnarray}  
with \(K\) be the Gaussian curvature of \(2\)-surfaces in the spacetimes and given by

\begin{eqnarray}\label{exe10}
\begin{split}
K&=\frac{1}{3}\rho-\frac{1}{2}\Pi+\frac{1}{4}\phi^2-\frac{1}{4}\left(\frac{2}{3}\Theta-\Sigma\right)^2\\
&=\frac{1}{3}\rho-\frac{1}{2}\Pi+\frac{1}{4}\phi^2-\frac{1}{4}\left(\frac{2}{3}\Theta-\Sigma\right)^2,
\end{split}
\end{eqnarray}  
where we have set \(\mathcal{E}=0\). Substituting the equations of \eqref{exe8} into \eqref{exe6} and \eqref{exe3}, and comparing the results we obtain the equation

\begin{eqnarray}\label{exe11}
\frac{1}{3}\Theta\left(\Theta+\frac{1}{2}\Sigma\right)=\rho+\frac{3}{2}\hat{\phi}+\frac{1}{4}\phi^2+\frac{1}{4}\Pi,
\end{eqnarray}  
which is the additional constraint we seek.

We also stress that additional constraints may be obtained by taking the dot derivatives  of the scalar equations obtained from the Gauss-Codazzi embedding equations.

As a limiting case, in shear-free spacetimes (\(\Lambda\) is necessarily zero), the constraints are greatly simplfied. In particular, the rotation satisfies \(\delta_{\mu}\Omega^{\mu}-\phi\Omega=0\), \(\lambda=(1/4)\Pi\) and the heat flux satisfies \(Q=-\left(2/3\right)\hat{\Theta}\). Hence, for the shear-free case we can make the following observation:

\begin{remark}
A hypersurface with Ricci tensor \eqref{pa1}, in a shear-free spacetime admitting a \(1+1+2\) decomposition is radiating if the expansion decreases along \(e^{\mu}\), is absorbing radiation if the expansion increases along \(e^{\mu}\), and neither radiates nor absorbs radiation if the expansion is constant along \(e^{\mu}\). In the increasing and decreasing cases we make the assumption that the geometry in the vicinity of the hypersurface is smooth.  
\end{remark}
The above remark can be seen in the study of horizon dynamics of black holes particularly in astrophysical and cosmological settings where, for example, there is in-falling radiation across the horizon which increases the horizon area, or a radiating black hole decreases the horizon area. 

Now, the local symmetry condition also implies that \(^3R_{\mu\nu}\) is Codazzi (by the second contracted Bianchi identity) and hence

\begin{eqnarray}\label{panel123}
D_{\sigma}\left(^3R_{\mu\nu}\right)-D_{\mu}\left(^3R_{\sigma\nu}\right)=0.
\end{eqnarray}
Writing \eqref{panel123} explicitly, and contracting with \(e^{\sigma}N^{\mu\nu},e^{\sigma}u^{\mu}u^{\nu},u^{\sigma}e^{\mu}e^{\nu}\), we obtain the following the system of equations

\begin{subequations}
\begin{align}
2\hat{\beta}-\left(\lambda-\beta\right)\phi&=0,\label{pene1}\\
\dot{\beta}&=0,\label{pene2}\\
\dot{\lambda}-2\dot{\beta}&=0,\label{pene4}
\end{align}
\end{subequations}
with \eqref{pene1} being identically satisfied as was shown previously. We see we also have that \(\dot{\lambda}=\dot{\beta}=0\) from \eqref{pene2} and \eqref{pene4}.

The function \(\lambda\), using the Ricci identities for \(e^{\mu}\), can be written as

\begin{eqnarray}\label{panel125}
-e^{\mu}e^{\gamma}N^{\nu\delta}R_{\mu\nu\delta\gamma}=\lambda=-2e^{\mu}N^{\nu\delta}D_{[\mu}D_{\nu]}e_{\delta},
\end{eqnarray}
where the term on the right hand side can be written entirely in terms of the covariant variables. 

The local symmetry condition further implies the following:

\begin{eqnarray}\label{mich1}
\begin{split}
D_{\mu}\left(^3R\right)&=0,\nonumber\\
\mbox{which implies}\ \ \dot{\lambda}+2\dot{\beta}=0,\ \ \hat{\lambda}+2\hat{\beta}&=0\ \ \mbox{and}\ \ \delta^2\left(\lambda+2\beta\right)=0.
\end{split}
\end{eqnarray}
so that \(\dot{R}=\hat{R}=0\), and \(\delta^2\left(\lambda+2\beta\right)=0\), where the first two conditions are satisfied. The covariant derivative of \(\lambda\) (using the left hand side of \eqref{panel125}) simplifies as

\begin{eqnarray}\label{panel127}
\begin{split}
D_{\sigma}\lambda&=-e^{\mu}e^{\gamma}N^{\nu\delta}D_{\sigma}\left(^3R_{\mu\nu\delta\gamma}\right)\\
&=0,
\end{split}
\end{eqnarray}
by \eqref{panel1}. We consequently have that \(\delta^2\lambda=0\Longrightarrow\delta^2\beta=0\). Indeed \eqref{pane7} is always satisfied. Hence, the metric \(h_{\mu\nu}\) is of constant scalar curvature.

We shall now make brief statements on the two cases of Proposition \ref{prop1} individually.

\subsection{The Einstein case}

Suppose we have that \(\lambda-\beta=0\) (with \(\xi\neq0,\phi\neq0\)). By using the vanishing of the dot derivatives of \(\lambda\) and \(\beta\), one can take derivatives of the scalar equations obtained from the Gauss-Codazzi embedding equations and show that the following equation has to be satisfied on the hypersurface:

\begin{eqnarray}\label{ei1}
\Theta^2\dot{\Sigma}=0,
\end{eqnarray}
so that either the expansion vanishes or the shear is constant along \(u^{\mu}\). If the expansion \(\Theta\) vanishes, then the function \(\lambda\) has to be proportional to the energy density \(\rho\) from \eqref{exe6}, in particular \(\lambda=(1/2)\rho\). We also have that the anisotropic stress \(\Pi\) vanishes, and the heat flux satisfies

\begin{eqnarray}\label{ei2}
Q=\frac{3}{2}\left(\Sigma\phi+\frac{1}{2}\hat{\Sigma}\right)+\left(\delta_{\mu}-2a_{\mu}\right)\Sigma^{\mu}.
\end{eqnarray}
Clearly in the shear-free case the hypersurface then models a conformally flat perfect fluid. Note in all this we will also consider that \eqref{exe2} will be satisfied. The condition \(\dot{\Sigma}=0\), which we willnot discuss here, can be used to obtain additional constraints on the hypersurface from the field equations.

\subsection{Case of vanishing twist and sheet expansion}

On the other hand, let us assume that the hypersurfaces are not Einstein and that \(\xi=\phi=0\). The following three constraints are required to be satisfied:

\begin{subequations}
\begin{align}
0&=\varepsilon^{\mu\nu}\left(\delta_{\mu}-a_{\mu}\right)\Sigma_{\nu},\label{ei3}\\
\delta_{\mu}\Omega^{\mu}&=\varepsilon^{\mu\nu}\left[\delta_{\mu}\Sigma_{\nu}+\Sigma_{\nu}^{\sigma}\left(\delta_{\mu}e_{\sigma}+\delta_{\{\mu}e_{\sigma\}}\right)\right],\label{ei4}\\
Q&=\frac{3}{4}\hat{\Sigma}-\frac{2}{3}\hat{\Theta}+\left(\delta_{\mu}-2a_{\mu}\right)\Sigma^{\mu}.\label{ei5}
\end{align}
\end{subequations}

\section{The hypersurfaces admitting a Ricci soliton structure}\label{soc6}

We now proceed to consider the cases for which the hypersurfaces considered in this work admit a Ricci soliton structure. We recall the definition of a Ricci soliton.

\begin{definition}\label{deef1}
A Riemannian manifold \(\left(\Xi,h_{\mu\nu},\varrho,X^{\mu}\right)\) is called a Ricci soliton if there exists a vector field \(X^{\mu}\) and a real scalar \(\varrho\in \mathbb{R}\) such that 

\begin{eqnarray}\label{pin22}
^3R_{\mu\nu}=\left(\varrho-\frac{1}{2}\mathcal{L}_{X}\right)h_{\mu\nu}.
\end{eqnarray}
where \(\mathcal{L}_{X}\) is the Lie derivative operator along the vector field \(X^{\mu}\), and \(\varrho\) is some constant from the set of real numbers.
\end{definition}
A Ricci soliton \(\left(\Xi,h_{\mu\nu},\varrho,X^{\mu}\right)\) is said to be shrinking, steady or expanding if \(\varrho>0,\varrho=0\) or \(\varrho<0\) respectively. The vector field \(X^{\mu}\) is sometimes referred to as the \textit{soliton field}, and we will sometimes use this terminology for the rest of the work.

With respect to the hypersurfaces considered in this work, it is clear that any Ricci soliton would be trivial since they are conformally flat, following from a well known result due to Ivey \cite{iv1}, Perelman \cite{gp1}, Petersen and Wylie \cite{pet1}, and Catino and Mantegazza \cite{cat1} etc., which can be formulated as\\
\ \\
\textit{Any nontrivial homogeneous Riemannian Ricci soliton must be non-compact, non-conformally flat, expanding and non-gradient,} \\
\ \\
Here by non-trivial, it is meant that the Ricci soliton is neither an Einstein space, nor is the product of an Einstein and a (pseudo)-Euclidean space. An immediate result that follows is the following
\begin{corollary}\label{theo2}
Let \(M\) be a \(4\)-dimensional spacetime, and let \(\Xi\) be a locally symmetric embedded \(3\)-manifold in \(M\), with Ricci tensor of the form \eqref{pa1}, and scalar curvature \(R\geq 0\). If \(\Xi\) admits a Ricci soliton structure, then, either \(\Xi\) is locally an Einstein \(3\)-space, or \(\Xi\) is locally isomorphic to

\begin{eqnarray}\label{eei1}
\mathcal{M}_E\times\mathbb{R},
\end{eqnarray}
where \(\mathcal{M}_E\) denotes a \(2\)-dimensional Einstein manifold.
\end{corollary}

Of course then any Ricci soliton - as it pertains to this work - of geometry \eqref{eei1}, must have vanishing sheet expansion and twist. 

That a Ricci soliton is trivial by no means imply uninteresting. In fact, these objects have a very rich structure and have been extensively studied. We will explore some of their properties with regards to our covariant approach.

Now, the use of the local symmetry condition on \eqref{pin22} gives

\begin{eqnarray}\label{pin35}
D_{\sigma}\mathcal{L}_{X}h_{\mu\nu}&=D_{\sigma}D_{(\mu}X_{\nu)}=0.
\end{eqnarray}

Upon comparing \eqref{pin22} with \eqref{pa1}, we see that the equations to to be solved by the hypersurfaces under consideration are

\begin{subequations}
\begin{align}
D_{(\mu}X_{\nu)}&=-\left[\left(\beta-\varrho\right)h_{\mu\nu}+\left(\lambda-\beta\right)e_{\mu}e_{\nu}\right],\label{pin36}\\
D_{\sigma}D_{(\mu}X_{\nu)}&=0.\label{pin37}
\end{align}
\end{subequations}
It is seen that \eqref{pin36} implies \eqref{pin37}, so we have just the equations \eqref{pin36} to solve. By choosing the general form of the vector field on the hypersurface as

\begin{eqnarray}\label{pin38}
X^{\mu}=\alpha e^{\mu}+m^{\mu},
\end{eqnarray}
confined to the hypersurface, where \(m^{\mu}\) is the component of \(X^{\mu}\) lying in the \(2\)-sheet and \(\alpha\in C^{\infty}\left(\Xi\right)\), we can expand \eqref{pin36} and contract with \(h^{\mu\nu},e^{\mu}e^{\nu}\) and \(u^{\mu}e^{\nu}\) to get the following set of equations

\begin{subequations}
\begin{align}
\hat{\alpha}+\delta_{\mu}m^{\mu}&=3\varrho+\lambda+2\beta,\label{pin39}\\
\hat{\alpha}+e_{\mu}\widehat{m^{\mu}}&=\varrho-\lambda,\label{pin40}\\
\dot{\alpha}&=0,\label{pin41}
\end{align}
\end{subequations}

We shall focus on the case where the vector field \(X^{\mu}\) is parallel to \(e^{\mu}\), i.e. \(m^{\mu}=0\). Subtracting \eqref{pin40} from \eqref{pin39} we obtain

\begin{eqnarray}\label{pin44}
\varrho=-\lambda-\beta.
\end{eqnarray}

In the case of vanishing scalar curvature, we have that \(\varrho=\beta\), and hence, the nature of the soliton is entirely specified by \(\beta\). Now, suppose \(^3R>0\). Then from the estimate \cite{ham2}

\begin{eqnarray}\label{pin45}
\frac{1}{3}\left(^3R\right)^2\leq |^3R_{\mu\nu}|^2\leq \left(^3R\right)^2,
\end{eqnarray}
it is straightforward to show that \(\beta\geq0\). To see this, the above estimate reduces to the two inequalities

\begin{eqnarray*}
\begin{split}
-\frac{2}{3}\left(\lambda-\beta\right)^2\leq0,\\
-2\beta\left(\beta+2\lambda\right)\leq0,
\end{split}
\end{eqnarray*}
the first of which is of course satisfied. The second can be rewritten as

\begin{eqnarray}\label{xxa1}
-4\beta\left(^3R-\frac{3}{2}\beta\right)\leq0.
\end{eqnarray}
Therefore, if \(\beta<0\), then one should have 

\begin{eqnarray*}
^3R\leq\frac{3}{2}\beta\Longrightarrow\  ^3R<0,
\end{eqnarray*}
contradicting \(^3R\) being strictly positive. Hence, we must have that \(\beta\geq0\). 

If \(\beta=0\), then the soliton is necessarily expanding since \(^3R=\lambda>0\). Consider the case \(\beta>0\). From \eqref{pin44}, whether the Ricci soliton is steady, shrinking or expanding will depend on the sign of the sum \(\lambda+\beta\): the Ricci soliton is steady, shrinking or expanding if \(\lambda+\beta\) is \(``=0", ``<0"\) or \(``>0"\) respectively. Notice that 

\begin{eqnarray}\label{yoyo}
\begin{split}
^3R=\lambda+2\beta&\geq0\\
\Longrightarrow\ \ \left(\lambda+\beta\right)+\beta&\geq0\\
\Longrightarrow\ \  -\varrho+\beta&\geq0\\
\Longrightarrow\ \ \varrho&\leq\beta,
\end{split}
\end{eqnarray}
with equality holding if and only if \(^3R=0\). Indeed if \(^3R>0\), then we have that \(\varrho<\beta\) (\(\varrho=\beta\Longrightarrow\ ^3R=0\)). From \eqref{exe7}, this gives the following bound on the scalar \(\varrho\): \(\beta-2\rho\leq\varrho\leq\beta\). Therefore whenever \(\beta\geq 2\rho\), \(\varrho\geq0\) and the soliton is non-expanding, this ensures that the energy density \(\rho\) is non-negative, which is desirable from a physical point of view. 

It follows that, for \(^3R=0\), the soliton is

\begin{itemize}
\item[1.] Steady implies \(\Xi\) is flat (\(\lambda=\beta=0\));

\item[2.] Shrinking implies \(\lambda<0\);

\item[3.] Expanding implies \(\lambda>0\).
\end{itemize}

Explicitly, write \eqref{pin44} as

\begin{eqnarray}\label{yop1}
\varrho=\left(\frac{2}{3}\Theta-\Sigma\right)\left(\frac{2}{3}\Theta+\frac{5}{4}\Sigma\right)-2\Omega^2-\frac{1}{4}\Pi+\frac{2}{3}\left(\rho+\Lambda\right).
\end{eqnarray}
The condition for the hypersurface to be of Einstein type can be expressed as

\begin{eqnarray}\label{yop2}
\frac{3}{4}\Sigma\left(\frac{2}{3}\Theta-\Sigma\right)-2\Omega^2+\frac{1}{4}\Pi=0,
\end{eqnarray}
So, for example, whenever \(X^{\mu}\) is a Killing vector for the metric \(h_{\mu\nu}\) of \(\Xi\) and \(\Xi\) is of Einstein type, one has that the following holds on \(\Xi\):

\begin{eqnarray}\label{yop3}
\left(\frac{2}{3}\Theta-\Sigma\right)\left(\frac{2}{3}\Theta+\frac{1}{2}\Sigma\right)-\frac{1}{2}\Pi+\frac{2}{3}\left(\rho+\Lambda\right)=0.
\end{eqnarray}

Indeed, it follows that, if a hypersurface \(\Xi\) admits a Ricci solition structure and on \(\Xi\) we have that 

\begin{eqnarray}\label{yop4}
\frac{3}{4}\Sigma\left(\frac{2}{3}\Theta-\Sigma\right)-2\Omega^2+\frac{1}{4}\Pi\neq0,
\end{eqnarray}
condition \(2.\) of Proposition \ref{prop1} holds, and \(\Xi\) has geometry \(\mathcal{M}_{E}\times\mathbb{R}\). For example, consider a shear-free spacetime with vanishing anisotropic stress. The condition \eqref{yop4} reduces to the requirement that the hypersurface must rotate. 

Also, noting that \(^3R\leq2\rho\), one requires that the following inequality must be satisfied on \(\Xi\):

\begin{eqnarray}\label{fgh}
\Omega^2\leq \frac{1}{3}\Theta^2.
\end{eqnarray}
Consider the case that the spacetime is expansion-free. Then the hypersurface cannot possibly rotate. The converse is of course possible: if the hypersurface is non-rotating, then it is possible to have \(0<\Theta^2\) for \(\Theta\) being non-zero.

Now let us return to the system \eqref{pin39} to \eqref{pin41}. For the solution of \(\alpha\), we can directly integrate \eqref{pin39} to obtain (we will take the constant of integration to be zero)

\begin{eqnarray}\label{pin49}
\begin{split}
\alpha&=\left(\varrho-\lambda\right)\chi\\
&=-\left[^3R+\left(\lambda-\beta\right)\right]\chi\\
&=-2\left(^3R-\frac{3}{2}\beta\right)\chi,
\end{split}
\end{eqnarray}

where \(\chi\) parametrizes integral curves of \(e^{\mu}\). In many instances when dealing with spacetimes of physical interest, the parameter \(\chi\) is can be identified with the radial coordinate. We will consider the interval \(0<\chi<\infty\). Indeed, for the Einstein case, the vector field \(X^{\mu}\) will point opposite the unit direction \(e^{\mu}\) (\(\alpha\neq0\Longrightarrow\ ^3R>0\Longrightarrow\) \(\Xi\) is non-flat). Of course then this does not accommodate the steady case.

In the non-Einstein case, for vanishing scalar curvature one has \(\alpha=3\beta\chi\), and since \(\alpha\neq0\), the soliton field points in the unit direction \(e^{\mu}\) if and only if the soliton is shrinking, and points opposite \(e^{\mu}\) if and only if the soliton is expanding. (Clearly \(^3R=0\) does not accommodate the steady case here as well,  since otherwise we would have that \(\alpha=0\).) Furthermore, it is not difficult to see that the Einstein case will necessarily have \(\varrho<0\).

Consider the case of positive scalar curvature. Recall that \(\beta\geq0\) in this case. If \(\beta=0\) we have an expanding soliton with soliton field pointing opposite \(e^{\mu}\). If \(\beta>0\), then, using \eqref{xxa1} we see that we have an expanding soliton with the soliton field points in the direction of \(e^{\mu}\). This leads us to state the following result:

\begin{proposition}\label{porp1}
Let \(M\) be a spacetime admitting a \(1+1+2\) decomposition, and let \(\Xi\) be a locally symmetric embedded 3-manifold in \(M\) with Ricci tensor of the form \eqref{pa1} with \(^3R\geq0\), admitting a Ricci soliton structure. Suppose the soliton field is non-trivial and parallel to the unit direction \(e^{\mu}\). If \(\Xi\) is Einstein, then \(\Xi\) is expanding opposite the direction of \(e^{\mu}\). If \(\Xi\) is non-Einstein, and the scalar curvature vanishes, then,

\begin{itemize}

\item \(\Xi\) is shinking in the direction of \(e^{\mu}\) for \(\beta>0\); or

\item \(\Xi\) is expanding opposite the direction of \(e^{\mu}\) for \(\beta<0\).

\end{itemize}
Otherwise, if the scalar curvature is strictly positive, then

\begin{itemize}

\item \(\Xi\) is expanding opposite the direction of \(e^{\mu}\) for \(\beta=0\); or

\item \(\Xi\) is expanding in the direction of \(e^{\mu}\) for \(\beta>0\).

\end{itemize}

\end{proposition}

One may consider more general cases of \eqref{pin38} where \(X^{\mu}\) also has a component along the \(u^{\mu}\) direction. Consider the vector field

\begin{eqnarray}\label{pin50}
X^{\mu}=\gamma u^{\mu}+\alpha e^{\mu} + m^{\mu}.
\end{eqnarray} 
Using \eqref{pin36} and contracting with \(h^{\mu\nu},e^{\mu}e^{\nu},u^{\mu}e^{\nu}\) and \(u^{\mu}u^{\nu}\) we obtain the following set of equations

\begin{subequations}
\begin{align}
\dot{\gamma}+\Theta\gamma+\hat{\alpha}+\delta_{\mu}m^{\mu}&=3\varrho+\lambda+2\beta,\label{pin51}\\
\hat{\alpha}+\left(\frac{1}{3}\Theta+\Sigma\right)\gamma+e_{\mu}\widehat{m^{\mu}}&=\varrho-\lambda,\label{pin52}\\
\dot{\alpha}-\hat{\gamma}&=0,\label{pin53}\\
\dot{\gamma}&=0.\label{pin54}
\end{align}
\end{subequations}
As before, if we consider the case of vanishing sheet component of \(X^{\mu}\), then comparing \eqref{pin54}, \eqref{pin52}, and \eqref{pin51} , \(\gamma\) can explicitly be written as

\begin{eqnarray}\label{pin55}
\left(\frac{2}{3}\Theta-\Sigma\right)\gamma=2\left(\varrho+\lambda+\beta\right). 
\end{eqnarray}
Interestingly, what this implies is that, if \(N^{\mu\nu}D_{\mu}u_{\nu}=(2/3)\Theta-\Sigma\) vanishes, then, whether the Ricci soliton is an expander, shrinker or steady does not depend on the choice of the component along \(u^{\mu}\). The analysis then follows as in the case of the former, with the solution for the component having an additional term in terms of \(\alpha\).  Consequently, for hypersurfaces of the form \(\chi=X\left(\tau\right)\) (with \(\tau\) parametrizing intregral curves of \(u^{\mu}\)), if the component along \(u^{\mu}\) is non-vanishing and the Ricci soliton is foliated by \(2\)-surfaces, this leads to an existence result with implications for black holes in spacetimes. Specifically, the Ricci soliton necessarily admits a marginally trapped tube structure which generalizes the notion of black hole boundaries. These are hypersurfaces foliated by 2-surfaces on which the trace of the second fundamental form with respect to the tangent to outgoing null geodesics (called the outgoing expansion null expansion and we denote this by \(\theta^+\)) vanishes \cite{ash1,ash2,ash3,ak1,ib1,boo2,ibb1,ib3,shef1,rit1,shef2}. In case of spacetimes admitting the splitting considered here, this is given by \cite{shef1,shef2}

\begin{eqnarray}
\theta^+=\frac{1}{\sqrt{2}}\left(\frac{2}{3}\Theta-\Sigma+\phi\right).
\end{eqnarray}

It is also easily seen from \eqref{pin55} that, if \(\gamma\neq0\), then 

\begin{eqnarray}\label{pin56}
\frac{2}{3}\dot{\Theta}-\dot{\Sigma}=0,
\end{eqnarray}
The function \(\alpha\) may now be solved for a given \(\gamma\). Notice in the case that \(\alpha=0\) (or constant), \(\gamma\) must be constant, and analysis follows just as the previous cases.

Let us now consider a special case where the vector field \(X^{\mu}\) is a generator of symmetries on \(\Xi\) and the ambient spacetime.

\subsection{\(X^{\mu}\) is a conformal Killing vector for the induced metric on the hypersurface}

Suppose \(X^{\mu}\) is a conformal Killing vector (CKV) for the metric on the hypersurface \(h_{\mu\nu}\). Then, there exists some smooth function \(\Psi\) on the hypersurface such that

\begin{eqnarray}\label{pin57}
\mathcal{L}_Xh_{\mu\nu}=2\Psi h_{\mu\nu}.
\end{eqnarray}
Taking the derivative of \eqref{pin57} and noting that \(D_{\sigma}\mathcal{L}_Xh_{\mu\nu}=0\), we have

\begin{eqnarray}\label{pin58}
h_{\mu\nu}D_{\sigma}\Psi=0,
\end{eqnarray}
and hence, \(\Psi\) must be constant in which case \(X^{\mu}\) is a homothetic Killing vector (HKV). The associated conformal factor \(\Psi\) can be found by setting

\begin{eqnarray}\label{pin59}
\Psi h_{\mu\nu}=-\left[\left(\beta-\varrho\right)h_{\mu\nu}+\left(\lambda-\beta\right)e_{\mu}e_{\nu}\right].
\end{eqnarray}
Taking the trace of \eqref{pin59} as well as contracting with \(e^{\mu}e^{\nu}\) gives respectively 

\begin{subequations}
\begin{align}
\Psi&=-\frac{1}{3}\left(^3R-3\varrho\right),\label{pin60}\\
\Psi&=-\left(\lambda-\varrho\right),\label{pin6001}
\end{align}
\end{subequations}
which upon equating gives

\begin{eqnarray}\label{pin61}
\frac{2}{3}\left(\lambda-\beta\right)=0.
\end{eqnarray}
Therefore the hypersurface must be of Einstein type. As expected, the vanishing of the derivative of \eqref{pin60} gives a constant scalar curvature.

Notice that from \eqref{pin44} one see that \(\varrho=-\left(\lambda+\beta\right)=-2\beta\), which gives the conformal factor as \(\Psi=-3\beta=-\ ^3R\). Also, \(\varrho=-2\beta\Longrightarrow\varrho\leq0\). This thus allows us to state the following

\begin{proposition}
Let \(M\) be a spacetime admitting a \(1+1+2\) decomposition, and let \(\Xi\) be a locally symmetric embedded 3-manifold in \(M\) with Ricci tensor of the form \eqref{pa1}, which admits a Ricci soliton structure. If the associated soliton field \(X^{\mu}=\gamma u^{\mu}+\alpha e^{\mu}\) is a conformal Killing vector for the induced metric on \(\Xi\), then \(X^{\mu}\) is a homothetic Killing vector with associated conformal factor given by

\begin{eqnarray}\label{pinew1}
\Psi=-3\left[2\rho-\frac{2}{3}\Theta^2+\Sigma^2+2\Omega^2\right],
\end{eqnarray}
and \(\Xi\) is a non-shrinking Ricci soliton of Einstein type. Furthermore, \(\Xi\) is steady if and only if \(\Xi\) is flat.

\end{proposition}

The above result also agrees with the well known fact that if the soliton field is a Killing field for the metric on \(\Xi\), then, \(h_{\mu\nu}\) is an Einstein metric.

\subsection{\(X^{\mu}\) is a conformal Killing vector for both metrics on the hypersurface and the ambient spacetime}  

Let \(X^{\mu}\) be a CKV for both the metric on the hypersurface and that on the ambient spacetime. Denote by \(\Psi\) and \(\bar{\Psi}\) the associated conformal factors respectively. The systems to be simultaneously solved simultaneously are \eqref{pin57} and

\begin{eqnarray}\label{pin62}
\mathcal{L}_Xg_{\mu\nu}=2\bar{\Psi} g_{\mu\nu}.
\end{eqnarray}
We can expand \eqref{pin62} as

\begin{eqnarray}\label{pin63}
\left(u_{(\mu}\nabla_{\nu)}+\nabla_{(\mu}u_{\nu)}\right)\gamma+\left(e_{(\mu}\nabla_{\nu)}+\nabla_{(\mu}e_{\nu)}\right)\alpha=\bar{\Psi}g_{\mu\nu},
\end{eqnarray}
(we are assuming here again the the vector \(X^{\mu}\) has no component lying in the \(2\)-sheet) from which we obtain the following set of equations:

\begin{subequations}
\begin{align}
\left(\frac{2}{3}\Theta-\Sigma\right)\gamma+\phi\alpha&=2\bar{\Psi},\label{pin64}\\
\dot{\gamma}+\mathcal{A}\alpha&=\bar{\Psi},\label{pin65}\\
\hat{\alpha}+\left(\frac{1}{3}\Theta+\Sigma\right)\gamma&=\bar{\Psi},\label{pin66}\\
\dot{\alpha}-\hat{\gamma}+\left(\frac{1}{3}\Theta+\Sigma\right)\left(\gamma-\alpha\right)&=0.\label{pin67}
\end{align}
\end{subequations}

We state and prove the following

\begin{proposition}\label{prop4}
Let \(M\) be a spacetime admitting a \(1+1+2\) decomposition, and let \(\Xi\) be a locally symmetric embedded 3-manifold in \(M\) with Ricci tensor of the form \eqref{pa1}, which admits a Ricci soliton structure. If the associated soliton field \(X^{\mu}=\gamma u^{\mu}+\alpha e^{\mu}\) is a conformal Killing vector for both \(h_{\mu\nu}\) and \(g_{\mu\nu}\), then either

\begin{enumerate}
\item \(\frac{1}{3}\Theta+\Sigma=0\); or

\item \(X^{\mu}\) is null, in which case, if \(X^{\mu}\) is a Killing vector for the metric \(g_{\mu\nu}\), then \(\Xi\) is flat, \(X^{\mu}\) is a Killing vector for the metric \(h_{\mu\nu}\), and the acceleration \(\mathcal{A}\) must vanish on \(\Xi\). And if \(\Xi\) is foliated by \(2\)-surfaces, then the Ricci soliton has the structure of a marginally trapped tube. 
\end{enumerate}
\end{proposition}
\begin{proof}
Comparing \eqref{pin53} to \eqref{pin67}, \eqref{pin54} to \eqref{pin65}, and subtracting \eqref{pin66} from \eqref{pin52}, we obtain respectively

\begin{subequations}
\begin{align}
\left(\gamma-\alpha\right)\left(\frac{1}{3}\Theta+\Sigma\right)&=0,\label{pin70}\\
\mathcal{A}\alpha&=\bar{\Psi},\label{pin71}\\
\varrho-\lambda&=\bar{\Psi}.\label{pin72}
\end{align}
\end{subequations}
From \eqref{pin70}, either \((1/3)\Theta+\Sigma=0\) or \(\gamma-\alpha=0\). Let us assume that \((1/3)\Theta+\Sigma\neq0\) and that \(\gamma-\alpha=0\). If \(X^{\mu}\) is a Killing vector for \(g_{\mu\nu}\), \(\bar{\Psi}=0\), and since \(\bar{\Psi}=\Psi\) (by \eqref{pin6001} and \eqref{pin72}), \(\Psi=0\) and \(X^{\mu}\) is a Killing vector for \(h_{\mu\nu}\). From \eqref{pin71} we have

\begin{eqnarray}\label{pin75}
\mathcal{A}\alpha=0.
\end{eqnarray}
Since \(\alpha\neq0\) (\(\alpha=0\Longrightarrow\gamma=0\Longrightarrow X^{\mu}=0\)), we must have that \(\mathcal{A}=0\). To show that \(\Xi\) is flat, first notice that as \(X^{\mu}\) is a KV for \(h_{\mu\nu}\), we know that \(\Xi\) is Einstein. From \eqref{pin72}, we have that \(\varrho=\lambda=\beta\), which upon comparing to \eqref{pin44} gives \(\lambda+2\beta=\  ^3R=0\). Hence, \(\beta=\lambda=0\Longrightarrow\ ^3R_{\mu\nu\delta\gamma}=0\). 

Now, the equation \eqref{pin64} can be written as

\begin{eqnarray}\label{pin73}
\left(\frac{2}{3}\Theta-\Sigma+\phi\right)\gamma=0.
\end{eqnarray}
Again, \(\gamma\neq0\) and hence we must have

\begin{eqnarray}\label{pin77}
\frac{2}{3}\Theta-\Sigma+\phi=0,
\end{eqnarray}
in which case \(2\)-surfaces in \(\Xi\) are marginally trapped. Therefore, as \(\Xi\) is foliated by \(2\)-surfaces, we have that \(\Xi\) has the structure of a marginally trapped tube.
\end{proof}

The following corollary follows from Proposition \ref{prop4}:

\begin{corollary}\label{cor2} 
Let \(M\) be a spacetime admitting a \(1+1+2\) decomposition, and let \(\Xi\) be a locally symmetric embedded 3-manifold in \(M\) with Ricci tensor of the form \eqref{pa1} with scalar curvature \(^3R\geq0\). Assume that \(\Xi\) admits a Ricci soliton structure with the soliton field \(X^{\mu}\) (we assume this vector has no component lying in the sheet) being a CKV for both \(h_{\mu\nu}\) and \(g_{\mu\nu}\). If \(X^{\mu}\) is not null, then for simultaneous non-vanishing of the expansion and rotation on \(\Xi\), the anisotropic stress cannot be zero.
\end{corollary} 

\begin{proof}
Since \(\gamma\neq\alpha\), we must have \((1/3)\Theta+\Sigma=0\). Direct substitution of \((1/3)\Theta+\Sigma=0\) into \eqref{yop2} (noting the \(\Xi\) is Einstein) gives

\begin{eqnarray}\label{cole1}
\Pi=\Theta^2+8\Omega^2.
\end{eqnarray}
The result then follows if \(\Omega\) and \(\Theta\) are not simultaneously zero.
\end{proof}

Under the assumptions of Proposition \ref{prop4}, hypersurfaces on which \((1/3)\Theta+\Sigma=0\) with the soliton field non-null, can admit non-flat Ricci soliton structure. 

The below proposition gives a non-trivial case where the soliton field can be explicitly found.

\begin{proposition}\label{cor3} 
Let \(M\) be a spacetime admitting a \(1+1+2\) decomposition, and let \(\Xi\) be a locally symmetric embedded 3-manifold in \(M\) with Ricci tensor of the form \eqref{pa1} with \(^3R\geq0\). Assume that \(\Xi\) admits a Ricci soliton structure with soliton field \(X^{\mu}=\gamma u^{\mu}+\alpha e^{\mu}\) being a CKV for both \(h_{\mu\nu}\) and \(g_{\mu\nu}\), and suppose \(X^{\mu}\) is non-null with the \(u^{\mu}\) component constant along \(e^{\mu}\). If the acceleration \(\mathcal{A}\) is covariantly constant and non-vanishing, then the components \(\alpha\) and \(\gamma\) have the general solutions  

\begin{subequations}
\begin{align}
\alpha&=h\left(\tau\right)e^{-\mathcal{A}\tau},\label{dcf1}\\
\gamma&=-\mathcal{A}\int h\left(\tau\right)e^{-\mathcal{A}\tau}d\tau,\label{dcf2}
\end{align}
\end{subequations}
for an arbitrary function \(h\left(\tau\right)\), with the solutions is subject to

\begin{eqnarray}\label{dds}
\phi h\left(\tau\right)e^{-A\tau}=\mathcal{A}\Theta\int h\left(\tau\right)e^{-\mathcal{A}\tau}d\tau.
\end{eqnarray}

\end{proposition} 

\begin{proof}

Since \(\gamma\neq\alpha\), we must have \((1/3)\Theta+\Sigma=0\). Noting \(\bar{\Psi}=0\), from combining \eqref{pin66} anb \eqref{pin67} we obtain the linear fiirst order partial differential equation

\begin{eqnarray}\label{gfe}
\dot{\alpha}+\hat{\alpha}-\mathcal{A}\alpha=0.
\end{eqnarray}
The above equation can be solved to give the general solution

\begin{eqnarray}\label{gfepi1}
\alpha&=h\left(\chi-\tau\right)e^{-\mathcal{A}\tau}
\end{eqnarray}
From \eqref{pin65} we have that 

\begin{eqnarray}\label{gfe2}
\gamma=-\mathcal{A}\int h\left(\chi-\tau\right)e^{-\mathcal{A}\tau}d\tau,.
\end{eqnarray}
However, by assumption, \(\gamma\) is independent of the parameter \(\chi\). Hence, setting \(h\left(\chi-\tau\right)=h\left(\tau\right)\) we obtain \eqref{dcf1} and \eqref{dcf2}. One can obtain \eqref{dds} by substituting \eqref{dcf1} and \eqref{dcf2} into \eqref{pin64} as well as using the fact that \((1/3)\Theta+\Sigma=0\).
\end{proof}
It immediately follows that

\begin{corollary}
Let \(M\) be an expansion-free spacetime admitting a \(1+1+2\) decomposition, and let \(\Xi\) be a locally symmetric embedded 3-manifold in \(M\) with Ricci tensor of the form \eqref{pa1} with \(^3R\geq0\). Assume that \(\Xi\) admits a Ricci soliton structure with soliton field \(X^{\mu}=\gamma u^{\mu}+\alpha e^{\mu}\) being a CKV for both \(h_{\mu\nu}\) and \(g_{\mu\nu}\), and suppose \(X^{\mu}\) is non-null with the \(u^{\mu}\) component constant along \(e^{\mu}\). If the acceleration \(\mathcal{A}\) is covariantly constant and non-vanishing, then the sheet expansion \(\phi\) must vanish.
\end{corollary}

\begin{proof}
If the assumptions herein hold, then we have the solutions \eqref{dcf1} and \eqref{dcf2} for the components along \(e^{\mu}\) and \(u^{\mu}\) respectively, subject to \eqref{dds}. But the spacetime is expansion-free, and hence from \eqref{dds} we have

\begin{eqnarray}\label{dds1}
\phi h\left(\tau\right)e^{-A\tau}=0.
\end{eqnarray}
Therefore, we have that either \(\phi=0\) or \(h\left(\tau\right)=0\). We rule out the latter as otherwise this would give \(\alpha=\gamma=0\Longrightarrow X^{\mu}=0\), and hence the result follows.
\end{proof}

If the function \(h\left(\tau\right)\) is strictly positive, then the converse of the above corollary also holds, i.e. under the assumptions of the corollary the vanishing of the sheet expansion implies the spacetime is expansion-free.

\section{Application to general locally rotationally symmetric spacetimes}\label{soc7}

We consider spacetimes with locally rotational symmetry \cite{cc1,gbc1}. 

\begin{definition}
A \textbf{locally rotationally symmetric (LRS)} spacetime is a spacetime in which at each point \(p\in M\), there exists a continuous isotropy group generating a multiply transitive isometry group on \(M\) (see \cite{crb1,gbc1} and associated references). The general metric of LRS spacetimes is given by

\begin{eqnarray}\label{jan29191}
\begin{split}
ds^2&=-A^2d\tau^2 + B^2d\chi^2 + F^2 dy^2 + \left[\left(F\bar{D}\right)^2+ \left(Bh\right)^2 - \left(Ag\right)^2\right]dz^2+ \left(A^2gd\tau - B^2hd\chi\right)dz,
\end{split}
\end{eqnarray}
where \(A^2,B^2,F^2\) are functions of \(\tau\) and \(\chi\), \(\bar{D}^2\) is a function of \(y\) and \(k\). The scalar (\(k\) can take the signs of negative, zero or positive, and fixes the geometry of the \(2\)-surfaces. \(k=-1\) corresponds to a hyperbolic \(2\)-surface, \(k=0\) corresponds to the \(2\)-plane, and \(k=+1\) corresponds to a spherical \(2\)-surface. The quantities \(g,h\) are functions of \(y\). 
\end{definition}
For the case \(g=h=0\) one has the LRS II class, a generalization of spherically symmetric solution to Einstein field equations (EFEs). Some other well known solutions of the LRS class include the G\"{o}del's rotating solution, the Kantowski-Sachs models and various Bianchi models. 

In these spacetimes, all vector and tensor quantities vanish. From \eqref{exe1}, we see that \(\Sigma=0\) (this can also be obtained from the fact that these spacetimes have vanishing cosmological constant, and since the cosmological constant is proportional to the square of the shear, the shear must vanish), and therefore all considerations in this section will be shear-free. (We have also noted the conformal flatness of the hypersurfaces and have set the magnetic and electric Weyl scalars to zero.) The field equations for these spacetimes can be obtained from the Ricci identities for the vector fields \(u^{\mu}\) and \(e^{\mu}\), as well as the contracted Bianchi identities. They are given by

\begin{itemize}

\item \textit{Evolution}

\begin{subequations}
\begin{align}
\frac{2}{3}\dot{\Theta}&=\mathcal{A}\phi- \frac{2}{9}\Theta^2 - 2\Omega^2 - \frac{1}{2}\Pi - \frac{1}{3}\left(\rho+3p\right),\label{subbe1}\\
\dot{\phi}&=\frac{2}{3}\Theta\left(\mathcal{A}-\frac{1}{2}\phi\right) +2\xi\Omega+ Q,\label{subbe2}\\
-\frac{1}{3}\dot{\rho}+\frac{1}{2}\dot{\Pi}&=-\frac{1}{6}\Theta\Pi+\frac{1}{2}\phi Q+\frac{1}{3}\Theta\left(\rho+p\right),\label{subbe3}\\
\dot{\xi}&=-\frac{1}{3}\Theta\xi + \left(\mathcal{A}-\frac{1}{2}\phi\right)\Omega,\label{sube3}\\
\dot{\Omega}&=\mathcal{A}\xi-\frac{2}{3}\Theta\Omega,\label{sube4}
\end{align}
\end{subequations}

\item \textit{Propagation}

\begin{subequations}
\begin{align}
\frac{2}{3}\hat{\Theta}&=2\xi\Omega + Q,\label{subbe4}\\
\hat{\phi}&=-\frac{1}{2}\phi^2 +\frac{2}{9}\Theta^2+2\xi^2-\frac{2}{3}\rho-\frac{1}{2}\Pi,\label{subbe5}\\
-\frac{1}{3}\hat{\rho}+\frac{1}{2}\hat{\Pi}&=-\frac{3}{4}\phi\Pi-\frac{1}{3}\Theta Q\label{subbe6}\\
\hat{\xi}&=-\phi\xi + \frac{1}{3}\Theta\Omega,\label{sube9}\\
\hat{\Omega}&=\left(\mathcal{A}-\phi\right)\Omega,\label{sube10}
\end{align}
\end{subequations}

\item \textit{Evolution/Propagation}

\begin{subequations}
\begin{align}
\hat{\mathcal{A}}-\dot{\Theta}&=-\left(\mathcal{A}+\phi\right)\mathcal{A}-\frac{1}{3}\Theta^2-2\Omega^2+\frac{1}{2}\left(\rho+3p\right),\label{subbe7}\\
\dot{\rho}+\hat{Q}&=-\Theta\left(\rho+p\right)-\left(2\mathcal{A}+\phi\right)Q,\label{subbe8}\\
\dot{Q}+\hat{p}+\hat{\Pi}&=-\left(\mathcal{A}+\frac{3}{2}\phi\right)\Pi-\frac{4}{3}\Theta Q-\left(\rho+p\right)\mathcal{A}.\label{subbe9}
\end{align}
\end{subequations}

\item \textit{Constraints}

\begin{subequations}
\begin{align}
0=\Omega Q,\label{gb1}\\
0=\left(\rho+p-\frac{1}{2}\Pi\right)\Omega+Q\xi,\label{gb2}\\
0=\left(2\mathcal{A}-\phi\right)\Omega.\label{subbe9}
\end{align}
\end{subequations}

\end{itemize}
Furthermore, for an arbitrary scalar \(\psi\) in a locally rotationally symmetric spacetime, the commutation relation for the dot and hat derivatives is given by \cite{cc1}

\begin{eqnarray}\label{ghh1}
\hat{\dot{\psi}}-\hat{\dot{\psi}}=-\mathcal{A}\dot{\psi}+\left(\frac{1}{3}\Theta+\Sigma\right)\hat{\psi}.
\end{eqnarray}

From \eqref{gb1} we have that either \(\Omega=0\) or \(Q=0\). We start by assuming that \(\Omega=0\) and \(Q\neq0\). Then, from \eqref{gb2} we must have \(\xi=0\), and hence the hypersurfaces to be considered have LRS II symmetries, and therefore we will treat them as embedded in the solutions of the LRS II class. We consider this below.

\subsection{LRS II class with non-vanishing heat flux: $\Omega=\xi=0$ and $Q\neq0$}

Firstly, the form of the Ricci tensor on the hypersurfaces can be expressed as 

\begin{eqnarray}\label{mj1}
^3R_{\mu\nu}=-\left(\hat{\phi}+\frac{1}{2}\phi^2\right)e_{\mu}e_{\nu}-\frac{1}{2}\left(\hat{\phi}+\phi^2-2K\right)N_{\mu\nu},
\end{eqnarray}
so we have

\begin{subequations}
\begin{align}
\lambda&=-\left(\hat{\phi}+\frac{1}{2}\phi^2\right),\label{mj2}\\
\beta&=-\frac{1}{2}\left(\hat{\phi}+\phi^2-2K\right),\label{mj3}
\end{align}
\end{subequations}
where 

\begin{eqnarray}\label{mj4}
K=\frac{1}{3}\rho-\frac{1}{2}\Pi+\frac{1}{4}\phi^2-\frac{1}{9}\Theta^2,
\end{eqnarray}
is the Gaussian curvature of the \(2\)-surfaces. The dot and hat derivatives of the Gaussian curvature are respectively given by

\begin{subequations}
\begin{align}
\dot{K}&=-\frac{2}{3}\Theta K,\label{mj5}\\
\hat{K}&=-\phi K.\label{mj6}
\end{align}
\end{subequations}
Using the fact that \(\dot{\lambda}=\dot{\beta}=\hat{\lambda}=\hat{\beta}=0\), we arrive at the following constraint equation

\begin{eqnarray}\label{mj7}
-\left(\frac{2}{3}\Theta+\phi\right)K=\frac{1}{2}\phi\left(\dot{\phi}+\hat{\phi}\right),
\end{eqnarray}
Notice that, whenever \(\phi\) vanishes, the constraint \eqref{mj7} says that either the immersed \(2\)-surfaces in the hypersurface are \(2\)-planes with \(K=0\) or, the hypersurface is expansion-free.

Furthermore, from \eqref{exe11} we have that

\begin{eqnarray}\label{mj8}
\frac{1}{3}\Theta^2=\rho+\frac{3}{2}\hat{\phi}+\frac{1}{4}\phi^2+\frac{1}{4}\Pi,
\end{eqnarray}
which simplifies as

\begin{eqnarray}\label{mj9}
\phi^2+\Pi=0.
\end{eqnarray}
(Notice how this forces the anisotropic stress to be non-positive.)

For LRS II class of spacetimes, \(\mathcal{H}=0\)  and \(\xi=0\) by definition, and hence, \eqref{exe2} is automatically satisfied. We shall now consider both cases of \ref{prop1}.

\subsubsection{Case \(1\): The Einstein case}

Let us assume that the hypersurface is Einstein. Consideration was given to this in Section \ref{soc5}. We show that case 1. implies case 2. of Proposition \ref{prop1}. 

Since the hypersurface is Einstein, using \eqref{mj2} and \eqref{mj3} gives \(\lambda-\beta=0\) as

\begin{eqnarray}\label{mj10}
K=-\frac{1}{2}\hat{\phi}.
\end{eqnarray}
Comparing \eqref{mj4} and \eqref{subbe5}, and using \eqref{mj9} we get \(\phi^2=0\), which gives \(K=0\) and therefore all \(2\)-surfaces as subspaces of the hypersurface are planes. 

\subsubsection{Case \(2\): The case of vanishing sheet expansion}

If we assume that the hypersurface is not of Einstein type, then \(\phi=0\). If the Gaussian curvature \(K\) is zero, then, the hypersurface is Einstein

The set of field equation \eqref{subbe1} to \eqref{subbe9} reduces to (we use here the fact that \(\phi=0\ \ \mbox{implies}\ \ \Pi=0\), as well as \(Q=-(2/3)\hat{\Theta}\))

\begin{subequations}
\begin{align}
\dot{\Theta}&=-\frac{1}{3}\Theta^2-\frac{1}{2}\left(\rho+3p\right),\label{mj11}\\
\hat{\Theta}&=0\iff Q=0,\label{mj12}\\
\dot{\rho}&=-\Theta\left(\rho+p\right),\label{mj13}\\
\hat{\rho}&=0,\label{mj14}\\
\hat{\mathcal{A}}&=-\left(\mathcal{A}^2+\frac{2}{3}\Theta^2\right),\label{mj15}\\
\hat{p}&=-\mathcal{A}\left(\rho+p\right),\label{mj16}
\end{align}
\end{subequations} 
coupled with the constraints

\begin{subequations}
\begin{align}
\mathcal{A}\Theta&=0,\label{mj17}\\
\rho&=\frac{1}{3}\Theta^2.\label{mj18}
\end{align}
\end{subequations} 

If \(K\neq0\), then \(\Theta=0\). But both \(\rho\) and \(p\) are zero, and hence

\begin{eqnarray}
\lambda-\beta=-K=0.
\end{eqnarray}
Therefore the hypersurface is flat. This allows for us to remark the following:

\begin{remark}
Any locally symmetric hypersurface in LRS II class of spacetimes, with Ricci tensor of form \eqref{pa1}, is necessarily flat.
\end{remark}

It indeed follows that any Ricci soliton structure admitted by the hypersurface is steady.
 
We saw that \(\Omega=0\ \ \mbox{implies}\ \ \xi=0\) if we assume \(Q\neq0\). Now let us assume that \(Q=0\) and \(\Omega\). For this we may consider two scenarios: the case where \(\Omega\neq0\) but \(\xi=0\) (these are the LRS I class of spacetimes which are stationary inhomogeneous), or the more general case for which \(\Omega\) and \(\xi\) are non-zero. 

\subsection{LRS I with vanishing heat flux: $\xi=0,\Omega\neq0$ and $Q=0$}

For these spacetimes the dot product of all scalars vanish (it was shown in \cite{ssgos1} that an arbitrary scalar \(\psi\) in general LRS spacetimes satisfy \(\dot{\psi}\Omega=\hat{\psi}\xi\), and hence for \(\xi=0\) and \(\Omega\neq0\), one has \(\dot{\psi}=0\)). 

Now, from \eqref{sube4} we have that

\begin{eqnarray}
0&=\Theta\Omega\label{mj19}.
\end{eqnarray}
From \eqref{mj19}  it is clear that \(\Theta\) must be zero. But an easy check shows that \eqref{fgh} cannot hold in this case, and hence these hypersurfaces cannot exist under the assumption of this work.

\subsection{LRS with vanishing heat flux: $\xi\neq0,\Omega\neq0$ and $Q=0$}

Let us now consider the case where both the rotation and the spatial twist are simultaneously non vanishing, noting the vanishing of the heat flux. If the anisotropic stress is zero then at least one of \(\Omega\) or \(\xi\) should vanish \cite{ve1}, and as such we can begin by assuming that \(\Pi\neq0\). Noting the expression \(\dot{\psi}\Omega=\hat{\psi}\xi\) we have the following set of equations \cite{ssgos1}

\begin{subequations}
\begin{align}
0&=\left(\Omega^2-\Sigma^2\right)+\frac{1}{3}\rho-\frac{1}{9}\Theta^2+\frac{1}{4}\Pi+\frac{1}{2}\mathcal{A}\phi,\label{mj20}\\
0&=-\left(\Omega^2-\Sigma^2\right)+\frac{1}{6}\left(\rho+3p\right)+\frac{1}{9}\Theta^2+\frac{1}{4}\Pi-\frac{1}{2}\mathcal{A}\phi,\label{mj21}\\
0&=\rho+p+\Pi.\label{mj22}
\end{align}
\end{subequations}
However, from \eqref{gb2} we have that 

\begin{eqnarray}\label{mj23}
\rho+p=\frac{1}{2}\Pi,
\end{eqnarray}
which gives \(\Pi=0\) by comparing to \eqref{mj22}. 
Hence, it follows that for these hypersurfaces to exist at least one of \(\Omega\) or \(\xi\) has to vanish. Therefore the only possibility for which the hypersurface can exist under the assumptions of this work is if \(\Omega=\xi=0\), in which case it was shown to be flat. We can summarize the result as follows:

\begin{proposition}
Let \(M\) be a locally rotationally symmetric spacetime. Any locally symmetric hypersurface in \(M\) orthogonal to the fluid velocity is necessarily flat. And if the hypersurface admits a Ricci structure, the soliton is steady with the components of the soliton field being constants. 
\end{proposition}

While the applications here do not yield non-flat geometries, a more general form of the Ricci tensor, when applied to the class of solutions in this section, will certainly demonstrate the applicability of the results we have obtained throughout this work.

\section{Summary and discussion}\label{soc8}

This work evolved out of an interest to employ the \(1+1+2\) and \(1+3\) covariant formalisms in the study of Ricci soliton structures on embedded hypersurfaces in spacetimes. As the geometry of Ricci solitons is well understood with a wealth of literature on the subject, there is potential application to the geometric classification of black hole horizons. 

In this work, we have carried out a detailed study of a particular class of locally symmetric embedded hypersurfaces in spacetimes with non-negative scalar curvature, admitting a \(1+1+2\) spacetime decomposition. This formalism has the advantage of bringing out the intricate details of the covariant quantities spacifying the spacetimes (or subsets thereof). As a first step, we prescribed the form of the Ricci tensor on the hypersurfaces and computed the associated curvature quantities. We computed the Ricci tensor for general hypersurfaces in \(1+1+2\) spacetimes and then specified the conditions under which the Ricci tensor for the general case reduces to that of the specified case we considered in this work.

First, we provided a characterisation of the hypersurfaces being considered. The locally symmetric condition implies conformal flatness. It is shown that a locally symmetric hypersurface embedded in a \(1+1+2\) spacetime, specified by the \eqref{pa1} is either an Einstein space or, the hypersurface is non-twisting with the sheet expansion vanishing. Properties of these cases were then briefly considered. The check for whether a hypersurface is Einstein or not reduces to a simple equation in a few of the matter and geometric variables. In particular, it reduces to whether or not 

\begin{eqnarray*}
\frac{3}{4}\Sigma\left(\frac{2}{3}\Theta-\Sigma\right)-2\Omega^2+\frac{1}{4}\Pi=0
\end{eqnarray*}
is satisfied. 

The components of the Ricci tensor were also shown to be covariantly constant, i.e. the hypersurface is of constant scalar curvature, and that the scalar curvature is bounded above by the energy density. Specifically, it was demonstrated that the scalar curvature has an upper bound as \(^3R\leq2\rho\), and hence \(0\leq\ ^3R\ \leq2\rho\). In essence we deal with metrics of bounded constant scalar curvature.

We then went on to consider the case in which a hypersurface admits a Ricci soliton structure. Solutions are considered for the case where the soliton field is parallel to the preferred spatial direction, as well as those also having components along the unit tangent direction orthogonal to the hypersurface. The nature of the soliton is determined by the eigenvalues of the Ricci tensor. This in turn determines the direction of the soliton field is the soliton field has component only along one of the preferred unit directions.

It was further considered the case in which the soliton field is a conformal Killing vector field for the induced metric on the hypersurface. It was shown that in this case the the soliton field is a homothetic Killing vector for the induced metric on the hypersurface, and that the hypersurface is of Einstein type. In the case that the Ricci scalar is strictly positive, the Ricci soliton is classified.

If the soliton field extends as a conformal Killing vector field to the metric of the ambient spacetime, then it was demonstrated that the quantity \((1/3)\Theta+\Sigma\) either vanishes or otherwise the soliton field is null. And if the soliton field is a Killing vector field, then the soliton was shown to be flat. The flat geometry in this case is a consequence of the soliton field being a Killing vector field, a well known fact. Otherwise, one could possibly have non-flat examples.

As another result for the case with \((1/3)\Theta+\Sigma=0\), it is shown that if the hypersurface is simultaneously rotating and expanding, the anisotropic stress cannot vanish on the hypersurface. A non-trivial case (in the sense that the type of soliton field is not specified and the hypersurface is not necessarily flat) was provided. It then followed that the sheet expansion will necessarily vanish on the soliton.

As one would expect, once a hypersurface admits a Ricci soliton structure, the geometry of the hypersurface is restricted and even more so when the choice of the soliton field is specified. We emphasize that, if one can find the soliton field and the constant specifying the nature of the soliton, the soliton equations have to be checked against the consistency of the field equations on the hypersurface as well when studying existence of Ricci soliton structure on embedded hypersurfaces. In general, this might not be possible, and possible only if the spacetime or class of spacetimes is specified.

A simple application was carried out against spacetimes with a high degree of symmetry, those exhibiting local rotational symmetry (LRS spacetimes). It turns out that the upper bound on the scalar curvature, expressed as a bound on the ratio of the rotation to the expansion scalar, places very strong constraints on the class of hypersurfaces considered in this work that can be admitted by LRS spacetimes. In particular it was shown that in this class of spacetimes, all hypersurfaces of type considered in this work is flat, and can be admitted by these spacetimes only if both rotation and spatial twist vanish simultaneously. And if they do admit a Ricci soliton structure the soliton will be steady, with the components of the soliton field being constants.

In subsequent works we seek to apply our approach to Ricci soliton structure on more general hypersurfaces in \(1+1+2\) decomposed spacetimes. Other possible extensions of this work could be studying Ricci soliton structure on general Lorentzian manifolds in the \(1+3\) covariant setting.

\section*{Acknowledgements}

AS acknowledges that this work was supported by the IBS Center for Geometry and Physics, Pohang University of Science and Technology, Grant No. IBS-R003-D1 and the First Rand Bank, through the Department of Mathematics and Applied Mathematics, University of Cape Town, South Africa. PKSD acknowledges support from the First Rand bank, South Africa. RG acknowledges support for this research from the National Research Foundation of South Africa.


\begin{thebibliography}{0}

\bibitem{ham1} 
R. Hamilton, 
The Ricci flow on surfaces,
{\it Contemp. Math.}, \textbf{71} (1988), 227.
 
\bibitem{df1} 
D. H. Friedan, 
Nonlinear modes in $2+\varepsilon$ dimensions, 
{\it Phys. Rev. Lett.}, \textbf{45} (1980), 1057.

\bibitem{df2} 
D. H. Friedan, 
Nonlinear modes in $2+\varepsilon$ dimensions, 
{\it Ann. Phys..}, \textbf{163} (1985), 318.

\bibitem{gp1} 
G. Perelman, 
The entropy formula for the Ricci flow and its geometric applications, 
\textit{arXiv}:math.DG/0211159 (2002).

\bibitem{gp2} 
G. Perelman,
Ricci flow with surgery on three-manifolds, 
\textit{arXiv}:math.DG/0303109 (2003).

\bibitem{gp3} 
G. Perelman, 
Finite extinction time for the solutions to the Ricci flow on certain three-manifolds, 
\textit{arXiv}:math.DG/0307245 (2003).

\bibitem{tom1} 
A. Tomimatsu and H. Sato,
Multi-Soliton Solutions of the Einstein Equation and the Tomimatsu-Sato Metric, 
{\it Prog. Theor. Phys. Suppl.}, \textbf{70} (1981), 215.

\bibitem{eric1} 
M. M. Akbar and E. Woolgar, 
Ricci Solitons and Einstein-Scalar Field Theory, 
{\it Class. Quantum Grav.}, \textbf{26} (2009), 055015.

\bibitem{bern1} 
J. Bernstein and T. Mettler, 
Two-dimensional gradient Ricci solitons revisited, 
{\it International Mathematics Research Notices}, (2013).

\bibitem{iv1} 
T. Ivey, 
Ricci solitons on compact three-manifolds, 
{\it Diff. Geom. Appl.}, \textbf{3} (1993), 301.

\bibitem{cao1} 
H.,-D. Cao, B.-L. Chen and X.-P. Zhu,
Recent developments on Hamilton’s Ricci flow,
{\it Surv. in Diff. Geom.}, volume \textbf{12} Int. Press, Somerville, MA, (2008).

\bibitem{bry1} 
R. L. Bryant,
Ricci flow solitons in dimension three with SO(3) symmetries,
{\it Preprint} available at www.math.duke.edu/bryant/, (2005).

\bibitem{bre1} 
R. Brendle, 
Rotational symmetry of Ricci solitons in higher dimensions, 
{\it Invent. Math.}, \textbf{194} (2013), 731.

\bibitem{cat1} 
G. Catino, P. Mastrolia and D. D. Monticelli,
Classification of expanding and steady Ricci solitons with integral curvature decay,
{\it Geom. Topol.}, \textbf{20} (2016), 2665.

\bibitem{der1} 
A. Deruelle, 
Steady gradient Ricci soliton with curvature in $L^1$, 
{\it Comm. Anal. Geom.}, \textbf{20} (2012), 31.

\bibitem{oc1} 
O. Chodosh,
Expanding Ricci solitons asymptotic to cones,
{\it Calc. Var. Part. Diff. Eq.}, \textbf{51} (2014), 1.

\bibitem{ma1} 
L. Ma, 
Expanding Ricci solitons with pinched Ricci curvature, 
{\it Kodai Math. J.}, \textbf{34} (2011), 140.

\bibitem{pet1} 
P. Petersen and W. Wylie,
Rigidity of gradient Ricci solitons,
{\it Pacific J. Math.}, \textbf{241} (2009), 329.

\bibitem{wal1} 
L. Ni and N. Wallach, 
On a classification of gradient Ricci solitons, 
{\it Math. Res. Lett.}, \textbf{15} (2008), 941.

\bibitem{pet2} 
P. Petersen and W. Wylie, 
On the classification of gradient Ricci solitons, 
{\it Geom. Topol.}, \textbf{14} (2010), 2277.

\bibitem{hac} 
G. S. Hall and M. S. Capocci, 
Classification and conformal symmetry in three-dimensional space-times, 
{\it J. Math. Phys.}, \textbf{40} (1999), 1466.

\bibitem{sea} 
F. C. Sousa, J. B. Fonseca and C. Romero, 
Equivalence of three-dimensional spacetimes, 
{\it Class. Quantum Grav.}, \textbf{25} (2008), 035007.

\bibitem{pg1} 
P. J. Greenberg,
The general theory of space-like congruences with an application to vorticity in relativistic hydrodynamics,
{\it J. Math. Anal. Appl.}, \textbf{30} (1970), 128.

\bibitem{cc1} 
C. Clarkson,
Covariant approach for perturbations of rotationally symmetric spacetimes,
{\it Phys. Rev. D}, \textbf{76} (2007), 104034.

\bibitem{crb1} 
C. A. Clarkson and R. K. Barrett, 
Covariant perturbations of Schwarzschild black holes,
{\it Class. Quantum Grav.}, \textbf{20} (2003), 3855.

\bibitem{gbc1} 
G. Betschart and C. A. Clarkson, 
Scalar field and electromagnetic perturbations on locally rotationally symmetric spacetimes,
{\it Class. Quantum Grav.}, \textbf{21} (2004), 5587.

\bibitem{gfr1} 
G. F. R. Ellis and M. Bruni, 
Covariant and gauge-invariant approach to cosmological density fluctuations, 
{\it Phys. Rev. D}, \textbf{40} (1989), 1804.

\bibitem{ge2} 
S. W. Hawking and G. F. R. Ellis, 
{\it The large scale structure of spacetime}, 
(Cambridge University Press, Cambridge, 1973).

\bibitem{al1} 
A. Lichnerowicz,
Courbure, nombers de Betti, et espaces sym\'{e}triques,
{\it Proceedings of the International Congress of Mathematicians, Cambridge, Mass.}, \textbf{2} (1950), 261.

\bibitem{kn1} 
K. Nomizu and H. Ozeki, 
A theorem on curvature tensor fields, 
\textit{Proc. Nat. Acad. Sci. USA}, \textbf{48} (1962), 206.

\bibitem{st1} 
S. Tanno,
Curvature tensor and covariant derivatives,
\textit{Ann. Math. Pura Appl.}, \textbf{96} (1972), 233.

\bibitem{hda1} 
H. Stephani, D. Kramer, M. A. H. MacCallum, C. Hoenselaers and E. Herlt, 
\textit{Exact solutions to Einstein's field equations}, 
(Second Edition, Cambridge University Press, Cambridge, 2003).

\bibitem{cdc1} 
C. D. Collinson and F. S\"{o}ler,
On the recurrency of a class of pseudo-Riemannian spaces,
\textit{Tensors (N.S.)}, \textbf{32} (1976), 87.

\bibitem{js1} 
J. Senovilla, 
Second-order symmetric Lorentzian manifolds: I. Characterization and general results, 
\textit{Class. Quantum Grav.}, \textbf{25} (2008), 245011.

\bibitem{ham2} 
R. Hamilton, 
Three-manifolds with positive Ricci curvature,
\textit{J. Diff. Geom.}, \textbf{17} (1982), 255.

\bibitem{ash1} 
A. Ashtekar and B. Krishnan, 
Dynamical horizons: energy, angular momentum, fluxes, and balance laws, 
\textit{Phys. Rev. Lett.}, \textbf{89} (2002), 261101.

\bibitem{ash2} 
A. Ashtekar and B. Krishnan,
Dynamical horizons and their properties, 
\textit{Phys. Rev. D}, \textbf{68} (2003), 104030.

\bibitem{ash3} 
A. Ashtekar and G. J. Galloway, 
Some uniqueness results for dynamical horizons, 
\textit{Advances in Theor. and Math. Phys.}, \textbf{9} (2005), 1.

\bibitem{ak1} 
A. Ashtekar and B. Krishnan, 
Isolated and dynamical horizons and their applications,
\textit{Liv. Rev. in Rel.}, \textbf{7} (2004), 10.

\bibitem{ib1} 
I. Booth and S. Fairhurst, 
Horizon energy and angular momentum from a Hamiltonian perspective, 
\textit{Class. Quantum Grav.}, \textbf{22} (2005), 4515.

\bibitem{boo2} 
I. Booth, 
Black hole boundaries, 
\textit{Can. J. Phys.}, \textbf{83} (2005), 1073.

\bibitem{ibb1} 
I. Booth and S. Fairhurst, 
Isolated, slowly evolving, and dynamical trapping horizons: geometry and mechanics from surface deformations, 
\textit{Phys. Rev. D}, \textbf{75} (2007), 084019.

\bibitem{ib3} 
I. Booth, L. Brits, J. A. Gonzalez and C. V. D. Broeck,
Marginally trapped tubes and dynamical horizons,
\textit{Class. Quant. Grav.}, \textbf{23} (2005), 413.

\bibitem{shef1} 
A. Sherif, R. Goswami and S. D. Maharaj, 
Some results on cosmological and astrophysical horizons and trapped surfaces,
\textit{Class. Quantum Grav.}, \textbf{36} (2019), 215001.

\bibitem{rit1} 
G. F. R. Ellis, R. Goswami, A. I. M. Hamid, and S. Maharaj,
Astrophysical black hole horizons in a cosmological context: Nature and possible consequences on Hawking radiation,
\textit{Phys. Rev. D}, \textbf{90} (2014), 084013.

\bibitem{shef2} 
A. Sherif, R. Goswami and S. D. Maharaj, 
Marginally trapped surfaces in null normal foliation spacetimes: A one step generalization of LRS II spacetimes,
\textit{Int. J. Geom. Meth. Mod. Phys.}, \textbf{17} (2020), 2150097.

\bibitem{ssgos1} 
S. Singh, G. F. R. Ellis, R. Goswami and S. D. Maharaj, 
New class of locally rotationally symmetric spacetimes with simultaneous rotation and spatial twist,
\textit{Phys. Rev. D}, \textbf{94} (2016), 104040.

\bibitem{ve1} 
H. V. Elst and G. F. R.  Ellis,
The covariant approach to LRS perfect fluid spacetime geometrie,
\textit{Class. Quant. Grav.}, \textbf{13} (1996), 1099.

\end{thebibliography}
\end{document}